\documentclass[conference]{IEEEtran}

\ifCLASSINFOpdf
  \usepackage[pdftex]{graphicx}
\else
\fi

\usepackage{hyperref}

\usepackage{comment}
\usepackage[linewidth=0.5pt]{mdframed}
\usepackage{multicol}
\let\labelindent\relax
\usepackage{enumitem}
\usepackage{subfig}
\usepackage[]{algorithm}
\usepackage[noend]{algpseudocode}
\usepackage[utf8x]{inputenc} 
\usepackage{multirow}
\usepackage{cite}
\usepackage{amsmath,amssymb,amstext, amsthm} % Lots of math symbols and environments
\usepackage{url}
\usepackage[normalem]{ulem}

\newtheorem{theorem}{Theorem}
\newtheorem{definition}{Definition}
\newtheorem{lemma}{Lemma}
\newcommand{\citeauthor}[1]{\textcolor{red}{Add authors}}

\usepackage[usenames,svgnames,table]{xcolor}
\newcounter{margin}
\colorlet{scolor}{DarkSlateBlue}
\colorlet{fcolor}{Green}

\newcommand*{\Fitpage}[1]{\resizebox{\linewidth}{!}{$#1$}}

\newcommand{\IPPE}{\text{IPPE}}
\newcommand{\IDX}{I}
\newcommand{\FHIPPE}{\text{FHIPPE}}

\newcommand{\OSSE}{\text{OSSE}}
\newcommand{\CLRZ}{\text{CLRZ}}
\newcommand{\genVec}{\textsf{genPoly}}
\newcommand{\genToken}{\textsf{genPred}}

\newcommand{\Bernoulli}[1]{\textbf{Bern}(#1)}
\newcommand{\Geometric}[2]{\textbf{Geo$_{#1}$}(#2)}
\newcommand{\Binomial}[2]{\textbf{Bi}(#1, #2)}

\newcommand{\TPR}[0]{\texttt{TPR}}
\newcommand{\FPR}[0]{\texttt{FPR}}

\newcommand{\countermax}{{\texttt{ctr}_\texttt{max}}}  % Max counter
\newcommand{\counteri}[1]{\texttt{ctr}_{#1}}   % Counter
\newcommand{\freqmax}{F_\texttt{max}}  % Max frequency of a kw
\newcommand{\sizemax}{S_\texttt{max}}  % Max kw per doc
\newcommand{\ap}[1][]{\Pi_{#1}}  % Access pattern
\newcommand{\aprand}[1][]{\tilde{\Pi}_{#1}}  % Obfuscated access pattern
\newcommand{\apsamp}[1][]{\pi_{#1}}  % Realization of aprand
\newcommand{\searchp}[1][]{\Phi_{#1}}
\newcommand{\token}[1][]{\tau_{#1}}
\newcommand{\tracerand}{\tilde{T}}
\newcommand{\nmescliser}{\#\text{tok}}
\newcommand{\nmessercli}{\#\text{doc}}
\newcommand{\toksize}{\tau_{\text{size}}}
\newcommand{\docsize}{D_{\text{size}}}
\newcommand{\overhead}[1]{\text{COMM-OVER}_{#1}}
\newcommand{\computation}[1]{\text{COMP-COMP}_{#1}}

\newcommand{\IKK}{\text{IKK}}
\newcommand{\IKKS}{\text{IKK}^*}

\newcommand{\secretkey}{\ensuremath{sk}}
\newcommand{\ciphertext}{\ensuremath{ct}}
\newcommand{\searchtoken}{\ensuremath{st}}

\begin{document}
\title{Obfuscated Access and Search Patterns in Searchable Encryption}

\author{\IEEEauthorblockN{Zhiwei Shang}
\IEEEauthorblockA{University of Waterloo\\
z6shang@edu.uwaterloo.ca}
\and
\IEEEauthorblockN{Simon Oya}
\IEEEauthorblockA{University of Waterloo\\
simon.oya@uwaterloo.ca}
\and
\IEEEauthorblockN{Andreas Peter}
\IEEEauthorblockA{University of Twente\\
a.peter@utwente.nl}
\and
\IEEEauthorblockN{Florian Kerschbaum}
\IEEEauthorblockA{University of Waterloo\\
florian.kerschbaum@uwaterloo.ca}}

\IEEEoverridecommandlockouts
\makeatletter\def\@IEEEpubidpullup{6.5\baselineskip}\makeatother
\IEEEpubid{\parbox{\columnwidth}{
		Originally published at:\\
    Network and Distributed Systems Security (NDSS) Symposium 2021\\
    21-24 February 2021, San Diego, CA, USA
}
\hspace{\columnsep}\makebox[\columnwidth]{}}

\maketitle

\begin{abstract}
Searchable Symmetric Encryption (SSE) allows a data owner to securely outsource its encrypted data to a cloud server while maintaining the ability to search over it and retrieve matched documents.
Most existing SSE schemes leak which documents are accessed per query, i.e., the so-called access pattern, and thus are vulnerable to attacks that can recover the database or the queried keywords.
Current techniques that fully hide access patterns, such as ORAM or PIR, suffer from heavy communication or computational costs, and are not designed with search capabilities in mind.
Recently, Chen et al.~(INFOCOM'18) proposed an obfuscation framework for SSE that protects the access pattern in a differentially private way with a reasonable utility cost.
However, this scheme always produces the same obfuscated access pattern when querying for the same keyword, and thus leaks the so-called search pattern, i.e., how many times a certain query is performed.
This leakage makes the proposal vulnerable to certain database and query recovery attacks.

In this paper, we propose $\OSSE$ (Obfuscated SSE), an SSE scheme that obfuscates the access pattern independently for each query performed.
This in turn hides the search pattern and makes our scheme resistant against attacks that rely on this leakage.
Given certain reasonable assumptions on the database and query distribution, our scheme has smaller communication overhead than ORAM-based SSE.
Furthermore, our scheme works in a single communication round and requires very small constant client-side storage.
Our empirical evaluation shows that $\OSSE$ is highly effective at protecting against different query recovery attacks while keeping a reasonable utility level.
Our protocol provides significantly more protection than the proposal by Chen et al.~against some state-of-the-art attacks, which demonstrates the importance of hiding search patterns in designing effective privacy-preserving SSE schemes.
\end{abstract}

\section{Introduction}
\label{section:intro}

Searchable Symmetric Encryption (SSE)~\cite{song2000practical} schemes allow a client to securely perform searches on an encrypted database hosted by a server. 
In a typical SSE scenario, the client first locally produces an encrypted version of the database and a search index, and outsources them to the server. 
Later, the client can issue queries for keywords that the server can securely run on the index, and retrieve the documents that match the query to decrypt them locally.

\ 

Even though the content of the queries and the documents are encrypted, during this interaction the server learns which documents are accessed, i.e., the \emph{access pattern}, and which queries are equal, i.e., the \emph{search pattern}. 
Most existing SSE schemes \cite{curtmola2011searchable, song2000practical, kamara2012dynamic, kamara2013parallel, cash2013highly, kamara2017boolean, bost2016ovarphiovarsigma} allow such leakage for performance considerations. 
However, recent studies \cite{islam2012access,cash2015leakage, zhang2016all} demonstrated that, with some prior knowledge of the outsourced database or a subset of the queries, an honest-but-curious server can recover the underlying keywords of queries with high accuracy, which violates the client's privacy.

There are different techniques that allow enhancing the privacy properties of SSE schemes, but they incur an utility cost which is typically a combination of communication overhead, extra computational complexity, and local client storage requirements.
Certain schemes, like those based on Oblivious RAM (ORAM)~\cite{goldreich1996software} or Private Information Retrieval (PIR)~\cite{chor1995private}, can fully hide the access pattern when reading a document from a database.
However, they incur a large communication and computation overhead, respectively, and are not specifically designed towards securely searching over an encrypted database (except for TWORAM~\cite{garg2016tworam}).
A recent framework by Chen et al.~\cite{chen2018differentially} protects access-pattern leakage in SSE by obfuscating the index of the database before outsourcing it.
This way, the server only learns obfuscated access patterns, making it harder to successfully carry out attacks on the client's privacy from such leakage.
However, despite its efficiency, this framework cannot hide \emph{search patterns} since the access pattern for each keyword is determined after outsourcing.
This search pattern leakage allows different practical attacks~\cite{islam2012access,liu2014search,pouliot2016shadow} to perform remakably well regardless of the access pattern obfuscation (see Sect.~\ref{section:evaluation}).

Motivated by this vulnerability of Chen et al.'s scheme~\cite{chen2018differentially}, in this work we propose $\OSSE$ (Obfuscated SSE), a new SSE scheme that protects \emph{both} the access and search patterns.
The main idea behind $\OSSE$ is that it produces a \emph{fresh obfuscation} per query, instead of just once when outsourcing the database, thus making it hard for the server to decide whether or not two queries are for the same keyword.
Our scheme allows to perform queries on the encrypted database and receive the matched documents \emph{in the same communication round} (TWORAM requires at least four rounds~\cite{garg2016tworam}).
Under some reasonable assumptions on the query and database distribution, $\OSSE$ achieves a lower communication overhead than TWORAM (e.g., only a small constant when the keyword distribution is uniform).
Our scheme relies on computation-heavy cryptographic techniques and thus its computational cost is considerable (e.g., it can require 30 minutes to run a query over $30\,000$ documents).
However, it requires very small constant client-side storage and the query process is parallelizable, which means that the server can process several queries in parallel to speed up the search, and return the results in a single round. 
Our experimental evaluation shows that obfuscating search patterns significantly reduces the performance of known practical query recovery attacks. 

Summarizing, our contributions are the following:
\begin{itemize}
 \item We propose $\OSSE$, an SSE that leverages Inner Product Predicate Encryption (IPPE)~\cite{shen2009predicate} whose main features are 
1) it obfuscates both access and search patterns; 
2) it runs queries in a single communication round, with an overhead that, depending on the database and query distribution, can be as small as a constant (three); 
3) its computational complexity is $O(n \log n/\log\log n)$, but can be significantly reduced with parallelization; 
and 4) $\OSSE$ requires a (very small) constant client-side storage, as opposed to ORAM-based SSE (TWORAM)~\cite{garg2016tworam} that requires $O(\log^2 n)$ storage.
 \item We instantiate the notion of differential privacy for queries and documents in SSE, which in turn imply search and access pattern privacy, respectively. We prove the differential privacy guarantees of $\OSSE$.
 \item We show that, even when $\OSSE$ is tuned to provide high utility (and a low differential privacy parameter $\epsilon$), our scheme still provides strong protection against four different query identification attacks \cite{islam2012access, liu2014search, cash2015leakage, pouliot2016shadow}. 
\end{itemize}

The paper is structured as follows. In Sect.~\ref{section:relwork} we review related work, before we introduce preliminaries in Sect.~\ref{section:preliminaries}.
We characterize the performance metrics we consider in Sect.~\ref{section:metrics} and present $\OSSE$ in Sect.~\ref{section:algorithm}.
We analyze the security, privacy, and complexity of our scheme in Sects.~\ref{section:security}, \ref{section:privacy}, \ref{section:complexity}, respectively, and evaluate it against empirical attacks in Sect.~\ref{section:evaluation}.
We conclude in Sect.~\ref{section:conclusions}.

\section{Related Work}
\label{section:relwork}

\subsection{Searchable Symmetric Encryption (SSE)}
SSE refers to a type of encryption that allows a data owner to outsource an encrypted database to an untrusted server while still preserving search functionalities. 
SSE was put forward by Song et al.~\cite{song2000practical}, who suggested several practical constructions whose search complexity is linear in the database size and secure under the Chosen Plaintext Attack (CPA). Goh et al.~\cite{goh2003secure} pointed out CPA was not adequate for SSE schemes.  
Curtmola et al.~\cite{curtmola2011searchable} provided formal notions of security and functionality for SSE, as well as the first constructions satisfying them with search complexity linear in the number of results (sub-linear in the size of the database). 
Subsequent works provided different security features, efficiency properties, and functionalities \cite{kamara2012dynamic, kamara2013parallel, cash2013highly, naveed2014dynamic, bost2016ovarphiovarsigma, kamara2017boolean}. However, all the above SSE schemes reveal which documents are accessed and returned in each query. 
This access pattern leakage opened the door to powerful query recovery attacks \cite{islam2012access, zhang2016all, cash2015leakage, pouliot2016shadow, blackstonerevisiting}.

\subsection{Query Recovery Attacks}
In 2012, Islam et al.~\cite{islam2012access} demonstrated that when knowing some statistics about a database and the content of a small fraction of queries, a semi-honest server could recover the contents of all queries with more than $90\%$ accuracy. 
This is the first powerful attack (known as IKK attack) utilizing access-pattern leakage. 
Subsequent works propose attacks that are effective when the adversary only knows a subset of the database~\cite{cash2015leakage, blackstonerevisiting} or has imperfect auxiliary information~\cite{pouliot2016shadow}. 
A different type of attack, called file-injection attack~\cite{zhang2016all, blackstonerevisiting}, showed that an active adversary could inject only a small number of carefully designed files in order to recover the content of queries by observing the access patterns of the injected files. 
Liu et al.~\cite{liu2014search} introduced an attack that leverages prior knowledge about the client's search habits and search-pattern leakage to identify the underlying keywords of the client's queries.
Even though these attacks aim at identifying the underlying keywords of the queries (query recovery), this in turn can allow the adversary to know the keywords of each document (database recovery).

\subsection{Oblivious RAM and Private Information Retrieval}
Oblivious RAM (ORAM), first introduced by Goldreich and Ostrovsky \cite{goldreich1996software}, was designed to hide memory access patterns by a CPU. 
Goldreich and Ostrovsky showed that a client could hide entirely the access patterns by continuous shuffling and re-encrypting data, with a $poly(\log n)$ communication overhead and $O(\log n)$ client-side storage. 
Since its proposal, there has been a fruitful line of research on further reducing this overhead \cite{stefanov2013path, apon2014verifiable, moataz2015constant, wang2015circuit, devadas2016onion}. 
Path ORAM \cite{stefanov2013path}, popular for its simplicity and efficiency, achieved $O(\log n)$ communication overhead with $O(\log n)$ client storage but required block size to be $\Omega(\log n)$.
Apon et al.~\cite{apon2014verifiable} showed that one can construct an ORAM scheme with constant communication overhead by leveraging fully homomorphic encryption (FHE). 
Even though subsequent works optimized this overhead~\cite{devadas2016onion,moataz2015constant}, FHE-based ORAM constructions still require a large communication overhead and rely on computationally expensive cryptographic primitives.
A $\Omega(\log n)$ lower bound of the overhead has been proven in the computational~\cite{larsen2018yes} and the statistical offline setting~\cite{goldreich1996software}. 
Naively combining SSE and ORAM might lead to a communication volume larger than directly downloading the entire database \cite{naveed2015fallacy}. 
Garg et al.~\cite{garg2016tworam} proposed TWORAM, an ORAM-based SSE scheme to hide access patterns with a communication overhead $O(\log n \cdot \log\log n)$, which requires at least four communication rounds and $O(\log^2 n)$ client-side storage.

Private Information Retrieval (PIR), proposed by Chor et al.~\cite{chor1995private}, allows a group of clients to privately retrieve documents from a public (unencrypted) database in the setting where there are many non-cooperating copies of the same database.
The main drawback of these schemes is that they require to touch every bit in the database per access, i.e., a computation overhead of $O(n)$.

We remark that ORAM and PIR provide private database access but, contrary to SSE schemes, they are \emph{not designed to provide search functionalities} by default (except for TWORAM~\cite{garg2016tworam}).
This means that a client must know beforehand the exact indices of the documents to be accessed.
Database search can also be implemented using only FHE~\cite{akavia2018secure, akavia2019setup}. 
The main drawback of these schemes is that, even though they have built-in search capabilities, they require multiple communication rounds to retrieve variable length results.

\subsection{Differentially Private SSE}
Chen et al.~\cite{chen2018differentially} proposed a differentially private obfuscation framework to mitigate access-pattern leakage in SSE.
The framework obfuscates the search index, i.e., a list of which documents contain keywords, by adding false positives and false negatives to it.
Consequently, the server only learns an obfuscated version of the access patterns, from which it is much harder to derive useful information.
However, since the obfuscation is fixed after outsourcing the index, querying for the same keyword multiple times results in the same obfuscated pattern.
This repetition actually reveals query frequencies of keywords, which can be used to recover the keywords being queried~\cite{liu2014search}.

\section{Preliminaries}
\label{section:preliminaries}

This section presents the formal definitions related to SSE and IPPE, that we need to characterize the leakage and security of our scheme, OSSE.
We give a high-level description of our scheme $\OSSE$ in Section~\ref{overview}.

\subsection{Searchable Symmetric Encryption (SSE)}

First, we introduce the notation related to SSE, which we summarize in Table~\ref{table:notation_database}. Let $\Delta = \{w_{(1)}, w_{(2)}, ..\dots, w_{(|\Delta| )}\}$ be the keyword universe, and let $|\Delta|$ denote its size.
Let $D$ be a document and let $id(D)$ be its document identifier, which we assume independent of its content. 
For notational simplicity, we treat documents as a list of the keywords they contain, e.g., $D=\{w_1,w_2,\dots,w_{|D|}\}$, and $w\in D$ means document $D$ contains keyword $w$. 
Let $\mathcal{D}$ be a dataset of $n$ documents, sorted by their ids.
$\mathcal{D}[i]$ is the $i$th document in the dataset; without loss of generality, we assume $id(\mathcal{D}[i])=i$.
Let $\mathcal{D}(w)$ be the list of documents that contain $w$ ordered by id. 

In the following, we define Searchable Symmetric Encryption (SSE) schemes and borrow the security definition~\cite{curtmola2011searchable}.

\begin{table}[t]
\centering
    \begin{tabular}{r p{6.5cm}}
        \textbf{Notation} & \textbf{Description}\\ \hline
        $\Delta$ & Keyword universe $\Delta=\{w_{(1)}, w_{(2)},\dots, w_{|\Delta|}\}$.\\
        $D$ & A document, $D=\{w_1, w_2,\dots,w_{|D|}\}$.\\
        $id(D)$ & Document identifier (or id) of $D$. \\
        $\mathcal{D}$ &  Dataset, list of all documents ordered by their ids. \\
        $\mathcal{D}[i]$ & $i$th document in $\mathcal{D}$, whose id is $i$.\\
        $n$ & Total number of documents, $n\doteq|\mathcal{D}|$.\\
        $\mathcal{D}(w)$ & List of all documents that contain $w$ ordered by id.\\
        $\IDX$ & Secure search index. \\% that the server uses to match queries to documents.\\
        \hline 
    \end{tabular}
    \caption{Database Notation}\label{table:notation_database}
\end{table}

\begin{definition}[Searchable Symmetric Encryption Scheme (SSE)]\label{def:SSE}
An SSE scheme is a collection of the following four polynomial-time algorithms: % (\textsf{Keygen}, \textsf{BuildIndex}, \textsf{Trapdoor}, \textsf{Search}) such that:
    \begin{itemize}
        \item $\textsf{Keygen}(1^\lambda)$ is a probabilistic key generation algorithm that is run by the client to setup the scheme. It takes a security parameter $\lambda$ and returns a secret key $\secretkey$ such that the length of $\secretkey$ is polynomially bounded in $\lambda$.
        \item $\textsf{BuildIndex}(\secretkey, \mathcal{D})$ is a (possibly probabilistic) algorithm that the client runs using the secret key $\secretkey$ and document collection $\mathcal{D}$ to generate an encrypted index $\IDX$ whose length is polynomially bounded in $\lambda$.
        \item $\textsf{Trapdoor}(\secretkey, w)$ is run by the client to generate a trapdoor $\tau_{w}$ for a keyword $w$ given the secret key $\secretkey$.
        \item $\textsf{Search}(\IDX,\tau_w)$ is run by the server in order to search for $\mathcal{D}(w)$. It takes an encrypted index $\IDX$ for a collection $\mathcal{D}$ and a trapdoor $\tau_w$ for keyword $w$ as inputs, and returns the set of identifiers of documents containing $w$.
    \end{itemize}
\end{definition}

In order to define the security of an SSE scheme, we borrow the concepts of \emph{history}, \emph{view} and \emph{trace} \cite{curtmola2011searchable}. First, the history contains the sensitive information that the client wants to keep private, namely the document plaintexts $\mathcal{D}$ and the sequence of keywords queried $\vec{w}$:

\begin{definition}[History] A history over $\mathcal{D}$ is a tuple
\begin{equation}
 H_t \doteq (\mathcal{D}, \vec{w})\,,\quad\text{where } \vec{w} \doteq (\vec{w}[1], \vec{w}[2], \dots, \vec{w}[t])\,.
\end{equation}
Here, $\vec{w}$ is the vector of underlying keywords of the $t$ queries.
\end{definition}
 
A \emph{partial history} of $H_t$, denoted $H_t^s$ for $s\leq t$, contains only the sequence of queries up to the $s$th query, i.e., $H_t^s\doteq(\mathcal{D}, \vec{w}')$ where  $\vec{w}' \doteq (\vec{w}[1], \dots, \vec{w}[s])$.

Next, the \emph{view} specifies what the server can see when running the SSE protocol, namely the document identifiers $id(\mathcal{D})$, the encryption of each document $\mathcal{E}(\mathcal{D}[i])$, a secure search index $\IDX$, and the query tokens (trapdoors) $\tau_i$:

\begin{definition}[View] The view of $H_t $ under key $\secretkey$ is the vector
\begin{equation}
 V_{\secretkey}(H_t) \doteq ( id(\mathcal{D}), \mathcal{E}(\mathcal{D}[1]),, ..., \mathcal{E}(\mathcal{D}[n]), \IDX, \tau_1, ..., \tau_t )\,,
\end{equation}
where $\tau_i$ is the query token for the $i$th query. The \emph{partial view} $V_{\secretkey}^s(H_t)$ only contains the tokens up to $\tau_s$, with $s\leq t$.
\end{definition}

Finally, the \emph{trace} of a history models the actual leakage of the SSE scheme. Before defining trace, we first define formally two types of leakage of SSE schemes: the access pattern and the search pattern.

\begin{definition}[Access Pattern]\label{def:ap}
The access pattern $\ap[\vec{w}]$ given a dataset $\mathcal{D}$ and a query vector $\vec{w}$ is a binary matrix of size $t \times n$ such that
\begin{equation}
\ap[\vec{w}][i, j] = \begin{cases}
1 &\text{ if } \vec{w}[i] \in \mathcal{D}[j]; \\
0 &\text{ otherwise}.
\end{cases}
\end{equation}
\end{definition}

\begin{definition}[Search Pattern]
The search pattern $\searchp[\vec{w}]$ given a dataset $\mathcal{D}$ and a query vector $\vec{w}$ is a symmetric binary matrix of size $t \times t$ such that
\begin{equation} 
\searchp[\vec{w}][i, j] = \begin{cases}
1 &\text{ if } \vec{w}[i] = \vec{w}[j];\\
0 &\text{ otherwise}.
\end{cases}
\end{equation}
\end{definition}

The access pattern reveals which documents contain which of the queried keywords, and the search pattern reveals which queries have the same underlying keyword.

\begin{definition}[\textsc{Trace}]\label{def:trace} 
The trace of $H_t$ is the vector
\begin{equation}
 T(H_t) \doteq (id(\mathcal{D}), |\mathcal{D}[1]|, \dots, |\mathcal{D}[n]|, \ap[\vec{w}], \searchp[\vec{w}])\,.
\end{equation}
\end{definition}

We are ready to define the adaptive semantic security for SSE. Informally, the definition states that an SSE scheme is secure if an adversary that observes the view can be simulated by an algorithm that only sees the trace; this implies that the trace contains all the information that is relevant to the adversary.

\begin{definition}[Adaptive Semantic Security for SSE~\cite{curtmola2011searchable}]
An SSE scheme is adaptively semantically secure if for all $t \in \mathbb{N}$ and for all (non-uniform) probabilistic polynomial-time adversaries $\mathcal{A}$, there exists a (non-uniform) probabilistic polynomial-time algorithm (the simulator) $\mathcal{S}$ such that for all traces $T_t$ of length $t$, all polynomially samplable distributions $\mathcal{H}_t$ over $\{ H_t : T(H_t) = T_t \}$ (i.e., the set of histories with trace $T_t$), all functions $f: \{0, 1\}^{|H_t|} \rightarrow \{ 0, 1 \}^{poly(|H_t|)}$, all $0 \leq s \leq t$, and sufficiently large $\lambda$:
\begin{equation*}
   \Fitpage{\big|Pr[ \mathcal{A}(V_{\secretkey}^s(H_t)) = f(H_t^s) ] - Pr[ \mathcal{S}(T(H_t^s)) = f(H_t^s) ] \big| < \delta(\lambda)}
\end{equation*}
where $H_t \leftarrow \mathcal{H}_t, \secretkey \leftarrow \textsf{Keygen}(1^\lambda)$, and probabilities are taken over $\mathcal{H}_t$ and the internal coins of $\mathcal{A}, \mathcal{S}$ and the underlying \textsf{Keygen, BuildIndex, Trapdoor, Search} algorithms. 
\end{definition}

Curtmola et al.~\cite{curtmola2011searchable} prove the equivalence of this definition and adaptive indistinguishability for SSE. 
In Section~\ref{section:security} we show that our SSE scheme provides adaptive semantic security.

\subsection{Inner Product Predicate Encryption (IPPE)}
\label{sec:FHIPPE}

We explain IPPE~\cite{katz2008predicate,shen2009predicate,bishop2015function}, a cryptographic tool that our SSE scheme uses.
In IPPE, plaintexts $x\in\Sigma$ and predicates $f\in\mathcal{F}$ are vectors $\Sigma=\mathcal{F}=\mathbb{Z}_N^s$.
We say that a plaintext $x$ satisfies a predicate $f$, denoted $f(x)=1$, if $\langle x, f \rangle=0$, where $\langle\cdot,\cdot\rangle$ denotes the inner product; otherwise, $f(x)=0$.

\begin{definition} \label{def:IPPE}
(\textsc{Symmetric-Key Inner Product Predicate Encryption}~\cite{shen2009predicate}) (IPPE) An IPPE scheme for the class of predicates $\mathcal{F}$ over the set of attributes $\Sigma $ consists of the following probabilistic polynomial-time algorithms.
	\begin{itemize}
      \item \textsf{Setup}$(1^\lambda )$ takes as input a security parameter $\lambda$ and outputs a secret key $\secretkey$.
      \item \textsf{Encrypt}$(\secretkey, x)$ takes as input a secret key $\secretkey$ and a plaintext $x \in \Sigma$ and outputs a ciphertext $\ciphertext_x$.
      \item \textsf{GenToken}$(\secretkey, f)$ takes as input a secret key $\secretkey$ and a predicate $f\in \mathcal{F}$ and outputs a search token $\searchtoken_f$.
      \item \textsf{Query}$(\searchtoken_f, \ciphertext_x)$ takes as input a token $\searchtoken_f$ for a predicate $f$ and a ciphertext $\ciphertext_x$ for plaintext $x$ and outputs $f(x)\in\{0,1\}$.
    \end{itemize}
\end{definition}

\textbf{Correctness}. 
IPPE correctnes requires that, for all $\lambda$, $x\in \Sigma $, and $f\in \mathcal{F}$; letting $\secretkey \leftarrow \textsf{Setup}(1^\lambda)$, $\searchtoken_f \leftarrow \textsf{GenToken}(\secretkey, \searchtoken_f)$, and $\ciphertext_x \leftarrow \textsf{Encrypt}(\secretkey, x)$,
	\begin{enumerate}
    \setlength\itemsep{0em}
    	\item If $\langle x, f \rangle = 0$, then $\textsf{Query}(\searchtoken_f, \ciphertext_x) = 1$.
        \item If $\langle x, f \rangle \not= 0$, then Pr[$\textsf{Query}(\searchtoken_f, \ciphertext_x) = 0$] $\geq 1 - \delta(\lambda )$ where $\delta $ is a negligible function.
    \end{enumerate}

An IPPE scheme can be used to securely evaluate polynomials \cite{katz2008predicate}, as follows. 
Let $\vec\alpha=(a_0, a_1, \dots, a_d)$ be the coefficients of a polynomial $P$ of degree $d$. 
The evaluation of this polynomial at point $x$ is simply $P(x) = \sum_{i = 0}^d a_i\cdot x^i = \langle \vec\alpha, \vec\beta \rangle$ where $\vec\beta = (x^0, x^1, \dots, x^d)$.
Therefore, the condition $P(x) = 0$ can be verified by checking whether \textsf{Query}$(\searchtoken_{\vec\beta}, \ciphertext_{\vec{\alpha}})=1$, where $\ciphertext_{\vec\alpha} \leftarrow \IPPE.\textsf{Encrypt}(\secretkey, \vec{\alpha})$ and $\searchtoken_{\beta}$ $\gets$ $\IPPE$.$\textsf{GenToken}$ ($\secretkey$, $\vec\beta)$.

\textbf{Security}. Shen et al.~\cite{shen2009predicate} define the security of an IPPE scheme by the following game $G$ between an adversary $\mathcal{A}$ and a challenger controlling the IPPE.
\begin{itemize}
    \item \textbf{Setup}: The challenger runs \IPPE.\textsf{Setup}$(1^\lambda)$, keeps $\secretkey$ to itself, and picks a random bit $b$.
    \item \textbf{Queries}: $\mathcal{A}$ adaptively issues two types of queries:
    \begin{itemize}
        \item [-] Ciphertext query. On the $j$th ciphertext query, $\mathcal{A}$ outputs two plaintexts $x_{j, 0}, x_{j, 1} \in \Sigma$. The challenger responds with $\IPPE.\textsf{Encrypt}(\secretkey, x_{j, b})$.
        \item [-] Token query. On the $i$th token query, $\mathcal{A}$ outputs descriptions of two predicates $f_{i, 0}, f_{i, 1} \in \mathcal{F}$. The challenger responds with $\IPPE.\textsf{GenToken}(\secretkey, f_{i, b})$.
    \end{itemize}
    $\mathcal{A}$'s queries are subject to the restriction that, for all ciphertext queries $(x_{j, 0}, x_{j, 1})$ and all predicate queries $(f_{i, 0}, f_{i, 1})$, $f_{i, 0}(x_{j, 0}) = f_{i, 1}(x_{j, 1})$.
    
    \item \textbf{Guess}: $\mathcal{A}$ outputs a guess $b'$ of $b$.
\end{itemize}
The advantage of $\mathcal{A}$ is defined as $Adv_{\mathcal{A}} = |Pr[b' = b] - \frac{1}{2}|$.

\begin{definition}
    (\textsc{Full Security} for IPPE~\cite{shen2009predicate}). A symmetric-key inner product predicate encryption scheme is fully secure if, for any probabilistic polynomial adversary $\mathcal{A}$, the advantage of $\mathcal{A}$ in winning the above game is negligible in $\lambda$.
\end{definition}

Roughly speaking, full security guarantees that given a set of tokens of predicates $f_1, \dots, f_k$ and a set of encryptions of plaintexts $x_1, \dots, x_t$, no adversary can gain any information about any predicate or any plaintext other than the value of each predicate evaluated on each of the plaintexts. 
The notion of predicate privacy is inherently impossible in the public-key setting, which is the reason why our construction works only in the private-key setting.

A stronger security notion in the context of $\IPPE$, called simulation-based security (SIM-security), requires that every efficient adversary $\mathcal{A}$ that interacts with the real $\IPPE$ can be simulated given only oracle access to the inner products between each pair of vectors that $\mathcal{A}$ submits to the real $\IPPE$. SIM-security implies full security, and an IPPE that provides SIM-secure is also called a Function-Hiding IPPE (FHIPPE).

\begin{definition}
    (\textsc{SIM-Security}  for IPPE). An IPPE is SIM-secure if it is fully secure and, for any efficient $\mathcal{A}$, there exists an efficient simulator $\mathcal{S}$ such that the following two games are computationally indistinguishable:
    \begin{equation*}
    \begin{array}{@{}l@{}c@{}l@{}}
%     \toprule
    \textsf{Real}_{\mathcal{A}}(1^{\lambda}): &\qquad& 
    \textsf{Ideal}_{\mathcal{A, S}}(1^{\lambda}): \\
    \quad \secretkey \gets \textsf{IPPE.Setup}(1^{\lambda}) &&
    \quad \secretkey' \gets \mathcal{S}.\textsf{Setup}(1^{\lambda})\\
    \quad b' \gets  G_{\mathcal{A}, \textsf{IPPE}}(1^{\lambda}) &&
    \quad b'' \gets  G_{\mathcal{A}, \mathcal{S}}(1^{\lambda})\\
    \quad \text{output }b' &&
    \quad \text{output }b''\\
%     \bottomrule
    \end{array}
    \end{equation*}
    Here, game $G_{\mathcal{A}, \textsf{IPPE}}$ represents the game defined previously, played between $\mathcal{A}$ and $\textsf{IPPE}$, and $G_{\mathcal{A, S}}$ is the game between $\mathcal{A}$ and $\mathcal{S}$.
\end{definition}

\section{Performance Metrics in SSE}
\label{section:metrics}

As explained earlier, existing SSE schemes offer different privacy, performance, and utility trade-offs. 
Schemes that leak the access and search patterns typically provide high performance and utility, while high protection solutions such as ORAM or PIR incur some computation cost, communication cost, or client storage requirements. 
Our goal is to find a middle-ground solution that obfuscates the search/access patterns while being lighter in terms of cost.
We make a distinction between \emph{utility metrics}, that directly affect the outcome of the protocol, and \emph{performance metrics}, that measure the cost of running the protocol.

From the client's perspective, the outcome of the protocol are the documents received from the server as a query response. Privacy-preserving schemes can sometimes return documents that do not contain the requested keyword (false positives), or miss out some documents that do contain the keyword (false negatives). This utility loss can be characterized by two metrics:
\begin{itemize}
	\item \textbf{True Positive Rate} (TPR): is the probability that the server returns a document that contains the queried keyword as a response to a query.
	\item \textbf{False Positive Rate} (FPR): is the probability that the server returns a document that \emph{does not} contain the requested keyword as a response to a query.
\end{itemize}

In most cases, the client aims for large TPR values (e.g., $>0.9999$) and low FPR values (e.g., $<0.01$).

On the other hand, we consider the following performance metrics that measure the cost of running the SSE algorithm:

\begin{itemize}
 \item \textbf{Communication Cost:} is the number of bytes exchanged between the client and the server and the number of communication rounds needed to perform each query.
 \item \textbf{Computational Cost:} is the number of operations that the client and server perform to run the SSE scheme.
 \item \textbf{Client Storage Cost:} is the amount of local memory that the client needs in order to run the SSE scheme.
\end{itemize}

The TPR and FPR can be palliated at the expense of an increase in the cost of the protocol:
The TPR can be increased by replicating documents on the server side, dividing each document into $K$ shards and requiring $k<K$ shards for document recovery~\cite{chen2018differentially}.
This incurs an extra communication, computation, and server storage cost.
Likewise, false positives can be trivially filtered by the client by checking whether or not a returned document contains the queried keyword. Thus, the effective cost of false positives is bandwidth consumption.

Our SSE scheme, $\OSSE$, requires a \emph{single communication round}, a (small) constant client storage, and can run queries in parallel, which can potentially allow for faster running times.

Finally, we consider both theoretical and empirical privacy metrics.
From the theoretical perspective, in Sect.~\ref{section:privacy} we study the performance of our scheme and previous proposals in terms of \emph{differential privacy}~\cite{dwork2008differential}, that has already been applied to our setting \cite{chen2018differentially}.
Then, in Sect.~\ref{section:evaluation} we use the \emph{empirical accuracy} of known attacks~\cite{islam2012access,liu2014search,cash2015leakage,pouliot2016shadow} to assess the privacy properties of SSE schemes in practice.

\section{Algorithm Description}
\label{section:algorithm}

We describe our proposal, that we call Obfuscated SSE ($\OSSE$).
Table~\ref{table:notation} summarizes the notation of this section.

\begin{table}[t]
\centering
    \begin{tabular}{r p{6.5cm}}
        \textbf{Notation} & \textbf{Description}\\ \hline
        $\ap[w]$ & Real access pattern for keyword $w$ (Def.~\ref{def:ap}).\\
        $\aprand[w]$ & Obfuscated access pattern for keyword $w$.\\
        $l$ & Document label, computed as a hash $l\doteq h(id(D))$.\\
        $\counteri{}$ & Counter used in our $\OSSE$ scheme.\\
        $\countermax$ & Maximum value of $\counteri{}$. \\
        $w_{\text{-1}}$ & A dummy keyword such that $w_{\text{-1}}\notin\Delta$.\\
        $I[i]$ & Encrypted polynomial that allows computing query matches for document $i$. \\
        $I$ & Search index, $I=(I[1],I[2],\dots,I[n])$.\\
        $\token[f]$ & Query token for a predicate $f$.\\
        $\Gamma_w$ & Set of $(\token[f],l)$ generated when querying for $w$.\\
        \hline 
    \end{tabular}
    \caption{Notation for the Algorithm Description}\label{table:notation}
\end{table}

\subsection{Construction Overview}\label{overview}

In a traditional SSE scheme, when the client queries for a keyword $w$, the server can compute the access pattern $\ap[w]$, i.e., a list of the identifiers of the documents in the database that contain $w$. 
The premise for our scheme is simple: we want to introduce random false positives (i.e., return documents that do not contain $w$) and false negatives (i.e., some of the documents that contain $w$ are not returned). 
By doing so, we only allow the server to observe an \emph{obfuscated access pattern} $\aprand[w]$. 
Each query results in a freshly random access pattern, which in turn also hides the search pattern, since given a set of obfuscated access patterns it is hard to tell whether or not these patterns where generated by queries for the same or different keywords. 
Despite the simplicity of this premise, achieving this functionality in a single communication round and without client storage is particularly challenging.

$\OSSE$ works as follows.
The client first generates its private key.
Then, for each document in the dataset, the client calls a function that we call $\genVec$ to generate a polynomial that encodes the keywords of that document.
These polynomials are encrypted using the private key, and this collection of ciphertexts forms the \emph{search index}.
The client sends the encrypted dataset and the search index to the server.
When querying for a keyword $w$, the client calls a function $\genToken(w)$ to generate a set of predicates, which are then converted into \emph{tokens}.
The client sends these tokens to the server.
The server evaluates the tokens on the search index (i.e., on the encrypted polynomials of each document); this evaluation returns a match for those documents that meet the query and must be returned to the client.
The key to adding random false positives and false negatives to this matching process lies in the construction of the functions $\genVec$ and $\genToken$, that we explain below.

In order to speed up computation, the client assigns a \emph{label} $l$ to each document by hashing its document identifier, i.e., the label of $D$ is $h(id(D))$. The hash function is publicly know, so the server can compute the label of each encrypted document. Then, the client attaches a label $l$ to each query token, so that the server only needs to evaluate each token on the encrypted polynomials of documents that share the label $l$.

\subsection{Polynomial Generation ($\genVec$)} 
\label{section:polynomial}

The function $\genVec(D,id(D))$ generates the polynomial coefficients for document $D$, that are later encrypted and sent to the server as part of the search index.
Before explaining this function, we clarify the following point: when the client sends a query token to the server, we want to avoid multiple document matches for a single token, since this reveals that two documents contain the same keyword and introduces correlations in the query response, which are hard to take into account in a privacy analysis. Therefore, the client designs the polynomials and query tokens so that a token can either match a single document or none at all. As a consequence, the client must generate multiple tokens per query in order to retrieve all the desired documents.

Now we explain how the client crafts the polynomials. Consider a document $D=\{w_1, w_2,\dots,w_{|D|}\}$ that contains $|D|$ keywords, has identifier $id(D)$, and label $l\doteq h(id(D))$. Let $\sizemax$ be the maximum number of keywords that any document can have. Then, the roots of the polynomial (i.e., the values for which it should return a match) are the following ($||$ denotes concatenation):
\begin{equation} \label{eq:polyroots}
\begin{aligned}
 (w_i||l||\counteri{w_i,l})\,,&\quad &&\text{for }i=1,2,\dots,|D|\,,\\
 (w_{\text{-1}}||0||0)\,,&\quad &&\text{for }i=|D|+1,\dots,\sizemax\,,\\
 (id(D)||0||\text{-1})\,.&&&
\end{aligned}
\end{equation}

The first set of roots ($w_i||l||\counteri{w_i,l}$) include each keyword $w_i$ in the document, the document label $l$, and a counter $\counteri{w_i,l}$ that is chosen so that none of the polynomials (one polynomial per document) have a root in common (so as to avoid more than one match per token). This counter starts at zero and keeps increasing as the client builds the search index. For example, two documents $\mathcal{D}[i]$ and $\mathcal{D}[k]$ might both contain a particular keyword $w$ and share a label $l=h(i)=h(k)$. However, their corresponding polynomials will not share a root, since the value in their $\counteri{w,l}$ field will be different. The maximum value of this counter is fixed before building the search index, and it is chosen such that, with overwhelming probability, the client does not need a higher counter value (more on this on Sect.~\ref{section:countermax}). These roots ($w_i||l||\counteri{w_i,l}$) allow the client to find true positive matches by generating tokens using this structure with their desired query keyword $w$, and looping through all labels $l$ and counter values.

The second set of roots ($w_{\text{-1}}||0||0$) include a dummy keyword $w_{\text{-1}}$ that is not in the keyword universe $w_{\text{-1}}\notin\Delta$. These roots act as padding, so that every polynomial has the same length (the same number of roots and coefficients) and thus hide the number of keywords that each document has.
Finally, the last root ($id(D)||0||\text{-1}$) allows the client to trigger false positives using only the document identifier $id(D)$ by generating a token with this structure.

Building the search index following \eqref{eq:polyroots} ensures that the only root in common between polynomials generated for different document is ($w_{\text{-1}}||0||0$), which is only used for padding and never queried. 
This means that a token can only trigger at most one match. 
When querying for keyword $w$, the client can generate false negatives by skipping some label/counter values in the generation of tokens ($w||l||\counteri{}$), force false positives by generating tokens from ($id(D)||0||\text{-1}$), and force non-matches by generating tokens from values that do not match any polynomial, e.g., ($w_{\text{-1}}||\text{-1}||0$). 

The function $\genVec$ takes as input a document $\mathcal{D}[i]$ and its label $i$, generates a polynomial based on its keywords as detailed in \eqref{eq:polyroots}, and returns the vector of polynomial coefficients $v_i$ (which is straightforward to compute from the roots). Since we have $\sizemax+1$ roots, the vector of polynomial coefficients will be of size $\sizemax+2$.

\subsection{Predicate Generation ($\genToken$)}
\label{section:tokgen}

The client calls $\genToken$ when performing a query. 
This function receives a keyword $w$ and outputs a set of (predicate, label) tuples.
The client converts them to (token, label) pairs and sends to the server to run the query.

The procedure to generate the predicate and label pairs is shown in Algorithm~\ref{alg:gentoken}.
The algorithm starts by initializing a multiset that will store the values on which to evaluate the polynomials, and their associated labels.
Then, the first part of the algorithm (lines \ref{alg1:part1}-\ref{alg1:part1end}) attempts to generate a query for each label and counter value, but only adds each with probability $p$.
Typically, $p$ will be close to 1. The goal of this part is to try to obtain a match in those documents that actually contain $w$, while having a small probability ($1-p$) of false negative.
The remaining code aims at hiding the true positives and non-matches produced by the first part of the algorithm.
The second part (lines \ref{alg1:part2}-\ref{alg1:part2end}) hides the true positives by generating false positives for each document in the dataset following a geometric distribution with parameter $1-q$ (we want $q$ close to zero to avoid a huge amount of false positives).
We explain this distribution choice in Sect.~\ref{section:geo}.
The third part (lines \ref{alg1:part3}-\ref{alg1:part3end}) hides the non-matches of the first part by generating non-matches for each label $l\in[|h|]$ following a geometric distribution with parameter $1-q$.
The remaining lines (\ref{alg1:part4}-\ref{alg1:part4end}) convert the predicates to a vector format such that its inner product with the document's polynomial coefficients returns whether or not the document should be returned. % compatible with $\FHIPPE$ and the previously generated polynomial coefficients.

\begin{algorithm}
    \caption{Predicate generation when querying for $w$}\label{alg:gentoken}
    \begin{algorithmic}[1]
    \Procedure{\genToken}{$w$}
    \State $\texttt{TupleSet} \gets \emptyset$ 
    \For{$l=1$ to $|h|$} \label{alg1:part1}
        \For{$\counteri{}=0$ to $\countermax$}
            \State with prob.~$p$, $\texttt{TupleSet}.add\big(  [ w\| l\| \counteri{}, l ]\big)$
        \EndFor
    \EndFor \label{alg1:part1end}
    \For{$id=1$ to $n$} \label{alg1:part2}       
				\State {$n_{FP}\gets \Geometric{}{1-q}$}
				\For{$k=1$ to $n_{FP}$}
					\State $\texttt{TupleSet}.add\big([id\| 0 \| \text{-1}, h(id)]\big)$
				\EndFor
    \EndFor \label{alg1:part2end}
    \For{$l=1$ to $|h|$} \label{alg1:part3}
				\State {$n_{NM}\gets \Geometric{}{1-q}$}
				\For{$k=1$ to $n_{NM}$}
					\State $\texttt{TupleSet}.add\big([w_{\text{-1}} \| \text{-1} \| 0, l ]\big)$
				\EndFor \label{alg1:part3end}
    \EndFor 
    
    \State $\texttt{PredLabelPairs} \gets \emptyset$  \label{alg1:part4}
    \For{$x, l \in \texttt{TupleSet}$} 
        \State {$\texttt{predicate}=\big(x^0, x^1, x^2, \dots, x^{\sizemax+1}\big)$}
        \State {$\texttt{PredLabelPairs}.add([ \texttt{predicate}, l ] )$}
    \EndFor
    \State return $\texttt{PredLabelPairs}$
    \EndProcedure \label{alg1:part4end}
    \end{algorithmic}
\end{algorithm}

Note that we can get a true positive if the inital coin flip succeeds (probability $p$) or if it fails but the geometric distribution generates at least one token, i.e., $\TPR=p+(1-p)q$. The probability of false positive is $\FPR=q$.

\subsection{Formal Interface}

We formalize $\OSSE$ interface following Definition~\ref{def:SSE} and using the IPPE functions in Definition~\ref{def:IPPE} as follows:

\begin{itemize}
    \item $\textsf{Keygen}(1 ^ \lambda)$: takes as input a security parameter $\lambda$ and returns $\FHIPPE.\textsf{Setup}(1^\lambda)$, which outputs secret key $\secretkey$.
    
    \item $\textsf{BuildIndex}(\secretkey, \mathcal{D})$: takes the secret key $\secretkey$ and the database $\mathcal{D}$, computes the polynomial coefficients $v_i=\genVec(D,id(D))$ for every document $D\in\mathcal{D}$, and then encrypts them by calling $\FHIPPE.\textsf{Encrypt}(\secretkey, v_i)$. The encrypted polynomials of all documents form the \emph{search index} $I$ of the database.
    
    \item $\textsf{Trapdoor}(\secretkey, w)$: takes the secret key $\secretkey$ and a keyword $w$, then calls $\genToken(w)$ to get a set of (predicate, label) tuples. Every predicate $f$ is transformed into a query \emph{token} $\token[f]$ using $\FHIPPE$.$\textsf{GenToken}(sk, f)$, and the output of $\textsf{Trapdoor(\secretkey, w)}$ is a set of (token, label) pairs $\Gamma_w$, that the client sends to the server.
    
    \item $\textsf{Search}(I, \Gamma_w)$:  takes as input the search index $I$ and a set of query tokens and labels $\Gamma_w$. Then, for every token and label $(\token[f], l) \in \Gamma_w$, it calls $\FHIPPE.\textsf{Query}(I[i], \token[f])$, for all document identifiers $i$ whose label is $l$; if the result of the inner product is $0$ (there has been a match), then it returns the corresponding document identifier $i$. The output of $\textsf{Search}(I, \Gamma_w)$ is the set of matched document identifiers $\vec{id}$.
The server then sends to the client the encrypted documents that correspond with these identifiers (ignoring possible duplicate ids in $\vec{id}$).
\end{itemize}

\subsection{Leakage Characterization}
\label{section:leakchar}

As explained above, $\OSSE$ hides the access and search patterns by adding random false positives, false negatives, and hiding the number of tokens that yield a non-match. In this section, we characterize the leaked access pattern, that we call \emph{obfuscated access pattern}. Then, we define the obfuscated trace, an update of Def.~\ref{def:trace} that characterizes the leakage of $\OSSE$. We use these concepts in the security and privacy analysis of $\OSSE$ in Sects.~\ref{section:security} and \ref{section:privacy}.

Since $\OSSE$ generates tokens independently per query, we study the adversary's observation when the user queries for a particular keyword $w$. 
First, the user calls $\genToken$ to generate a sequence of tokens and labels $\Gamma_w$, and sends them to the server. 
The server evaluates each of the tokens $(\token,l)\in\Gamma_w$ using the search index and either obtains a single match (in a document that has a label $l$) or a non-match. 
Therefore, for each token and label pair $(\token,l)$ there are $n+|h|$ possible adversary observations: either document $\mathcal{D}[i]$ (with label $h(id(\mathcal{D}[i]))=l$) satisfied the token ($n$ possibilities), or there was a non-match ($|h|$ possibilities, one for each label). 
The adversary observes one of these outcomes for each single token received. Therefore, the obfuscated access pattern $\aprand[w]$ for this query can be seen as a multi-nary $n+|h|$ vector. 
We characterize this vector now, following Alg.~\ref{alg:gentoken}.

First, let's study the value of $\aprand[w][i]$ when $i\in[n]$, i.e., the number of times the adversary observes a match for the $i$th document. If this document contains $w$, then lines \ref{alg1:part1}-\ref{alg1:part1end} will generate a true match with probability $p$. Otherwise, these lines will not generate a match for document $i$. Then, lines \ref{alg1:part2}-\ref{alg1:part2end} add false positive matches for each document following a geometric distribution with parameter $1-q$. The rest of the algorithm just generates extra non-matches. Therefore,
\begin{equation} \label{eq:dist1}
 \aprand[w][i]\sim\begin{cases}
      \Bernoulli{p}+\Geometric{}{1-q} &\text{if } w\in \mathcal{D}[i]\,,\\
      \Geometric{}{1-q} &\text{if } w\notin \mathcal{D}[i]\,,
     \end{cases}\quad \text{for } i\in[n]\,.
\end{equation}

Now, let's study the number of non-matches the adversary sees with label $l\in[|h|]$, i.e., $\aprand[w][n+l]$.
In lines \ref{alg1:part1}-\ref{alg1:part1end}, every token generated with label $l$ will return a non-match except for those documents with label $l$ that contain $w$.
Since each token is generated with probability $p$, the non-matches with label $l$ follow a binomial distribution with parameters $g_l\doteq\countermax-|\mathcal{D}(w)_l|$ and $p$, where $\mathcal{D}(w)_l$ is the set of ids of documents with label $l$ that contain keyword $w$.
Additionally, lines \ref{alg1:part3}-\ref{alg1:part3end} generate a number of non-matches with label $l$ following a geometric distribution with parameter $1-q$.
Therefore,
\begin{equation} \label{eq:dist2}
 \aprand[w][n+l]\sim \,\,\Binomial{g_l}{p}+\Geometric{}{1-q}\quad\text{for } l\in[|h|]\,.
\end{equation}

We are ready to define the obfuscated access pattern:
\begin{definition}[Obfuscated Access Pattern]
The obfuscated access pattern $\aprand[\Vec{w}]$ over a history $H_t$ is a multi-nary matrix of size $t\times n+|h|$ whose $i$th row $\aprand[{\Vec{w}[i]}]$ is characterized by \eqref{eq:dist1} and \eqref{eq:dist2}.
\end{definition}

The obfuscated trace summarizes the leakage of $\OSSE$:
\begin{definition}[Obfuscated Trace]\label{def:obfuscatedtrace}
The trace of a history $H_t$ when searching for $t$ keywords $\Vec{w}$ in $\OSSE$, called obfuscated trace and denoted by $\tracerand$, is
\begin{equation}
 \tracerand(H_t)\doteq(id(\mathcal{D}), |\mathcal{D}[1]|,\dots,|\mathcal{D}[n]|,h, \freqmax, \sizemax,\aprand[\Vec{w}])\,,
\end{equation}
where $\aprand[\Vec{w}]$ is the \emph{obfuscated access pattern} for the query vector $\vec{w}$, $h$ is the hash function used to generate document labels, and $\freqmax$ and $\sizemax$ are the maximum number of documents that contain a particular keyword and the maximum number of keywords per document, respectively.
\end{definition}

\subsection{Choice of Geometric Distribution for False Positives}
\label{section:geo}

As we explain above, $\OSSE$ adds false positives following a geometric distribution.
The reason for this choice is the following: when the adversary observes $k>1$ matches for a certain document $\mathcal{D}[i]$, they know for sure that at least $k-1$ are false positives (because there can only be a single true positive).
However, since the geometric distribution is memoryless, the number of matches $k$ does not reveal anything more than just observing a single match about whether a keyword $w$ is in $\mathcal{D}[i]$ or not.
Mathematically,
\begin{equation}
  \frac{\Pr(\aprand[w][i]=1|w\in \mathcal{D}[i])}{\Pr(\aprand[w][i]=1|w\notin \mathcal{D}[i])}=\frac{\Pr(\aprand[w][i]=k|w\in \mathcal{D}[i])}{\Pr(\aprand[w][i]=k|w\notin \mathcal{D}[i])}\,,
\end{equation}
for all $k>1$.
Additionally, as we prove in Sect.~\ref{section:privacy}, the geometric distribution allows $\OSSE$ to provide differential privacy guarantees.

\subsection{Label Generation Function $h$ and $\countermax$.}
\label{section:countermax}

Our protocol assigns a label to each document by hashing its document identifier using a function $h(\cdot)$. Then, when running $\genVec$ while building the index ($\textsf{BuildIndex}$), $\OSSE$ uses a counter to ensure that no token can match more than one document. The size of the label (denoted $|h|$) determines how large the counter value must be: a small $|h|$ will cause a lot of documents with a particular keyword $w$ to have the same label $l$, and thus will make $\counteri{w,l}$ large. However, the maximum value of $\counteri{w,l}$, denoted $\countermax$, must be fixed before outsourcing the database, otherwise the index building process would leak information about keyword frequency per label. If the client commits to a maximum value $\countermax$ that is too small, it will not be possible to generate the coefficients such that their roots are all unique unless some keywords are removed from some documents. We call this a $\textsf{BuildIndex}$ \emph{failure}. In order to avoid a failure and ensure that $\textsf{BuildIndex}$ \emph{succeeds}, we set $|h|=\freqmax$ and $\countermax=3\ln n/\ln\ln \freqmax$. The following theorem guarantees an overwhelming success probability with these parameters:

\begin{theorem} Let $h$ be a hash function that outputs values uniformly at random in the range $|h|=\freqmax$, i.e., $h:[n]\to[\freqmax]$. Let $\countermax=3\ln n/\ln\ln\freqmax$. Then, the index building process \emph{succeeds} with probability $\geq 1-1/n$. 
\end{theorem}

\begin{proof}
 Refer to Appendix~\ref{app:countermax1}.
\end{proof}

In practice, $\countermax$ can be even smaller. In our experiments in Sect.~\ref{section:evaluation}, with a real dataset with $n\approx 30\,000$ documents and $\freqmax\approx2\,000$, a value $\countermax=7$ is enough to successfully build the index.

\section{Security}
\label{section:security}

In this section, we prove that $\OSSE$ is adaptively semantically secure when the underlying $\FHIPPE$ is SIM-secure. For this, we show that a simulator that receives the obfuscated trace (Def.~\ref{def:obfuscatedtrace}) can simulate the view of an adversary who is allowed to adaptively issue queries.

\begin{theorem} $\OSSE$ is adaptively semantically secure.
\end{theorem}
\begin{proof} Let $\mathcal{S}$ be a simulator that sees a partial obfuscated trace $\tracerand(H_t^s)$. We show that this simulator can generate a view $(V_t^s)^*$ that is indistinguishable from the actual partial view of the adversary $V_{sk}^s(H_t)$ for all $0\leq s\leq t$,  all polynomial-bounded functions $f_p$, all probabilistic polynomial-time adversaries $\mathcal{A}$, and all distributions $\mathcal{H}_t$; except with a negligible probability, where $t \in \mathbb{N}, H_t \xleftarrow{R} \mathcal{H}_t $ and $sk \leftarrow \textsf{Keygen}(1^k)$.

We recall the notions of (partial) view and (partial) obfuscated trace in $\OSSE$:
\begin{align*} 
 V_{sk}(H_t) &\doteq (id(\mathcal{D}), \mathcal{E}(\mathcal{D}[1]), \dots, \mathcal{E}(\mathcal{D}[n]), \IDX, \Gamma_1, \dots, \Gamma_t  )\\
 V_{sk}^s(H_t)&\doteq (id(\mathcal{D}), \mathcal{E}(\mathcal{D}[1]), \dots, \mathcal{E}(\mathcal{D}[n]), \IDX, \Gamma_1, \dots, \Gamma_s  )\\
 \tracerand(H_t)&\doteq(id(\mathcal{D}), |\mathcal{D}[1]|..|\mathcal{D}[n]|,h, \freqmax, \sizemax,\aprand[{\vec{w}[1]}]..\aprand[{\vec{w}[t]}])\\
 \tracerand(H_t^s)&\doteq(id(\mathcal{D}), |\mathcal{D}[1]|..|\mathcal{D}[n]|,h, \freqmax, \sizemax,\aprand[{\vec{w}[1]}]..\aprand[{\vec{w}[s]}])
\end{align*}

First, note that simulating $id(\mathcal{D})$ is trivial since we assumed that $id(\mathcal{D})=(1,\dots,n)$. Also, encrypted documents $\mathcal{E}(\mathcal{D}[i])$ are indistinguishable from a random string $e_i^* \xleftarrow{R} \{0, 1\}^{|\mathcal{D}[i]|} $ since $\mathcal{E}$ is semantically secure. However, simulating the index $\IDX$ and the search token an label lists $\Gamma_i$ is non-trivial. 

Since the underlying $\FHIPPE$ is SIM-secure, there exists a simulator $\mathcal{S}_{\FHIPPE} $ which can simulate $\FHIPPE$, namely $\mathcal{S}_{\FHIPPE}$ has access to an encryption oracle $\mathcal{S}_{\FHIPPE}$.$\textsf{Encrypt}$ and a token generation oracle $\mathcal{S}_{\FHIPPE}$.$\textsf{GenToken}$ such that the outputs of the two oracles are indistinguishable from those of $\FHIPPE$.$\textsf{Encrypt}$ and $\FHIPPE$.$\textsf{GenToken}$ under the restriction that the inner product predicate of each pair of input to $\FHIPPE$.$\textsf{Encrypt}$ and $\FHIPPE$.$\textsf{GenToken}$ is equal to that of each pair of input to $\mathcal{S}_{\FHIPPE}$.$\textsf{Encrypt}$ and $\mathcal{S}_{\FHIPPE}$. $\textsf{GenToken}$.

$\bullet$ \textbf{Simulate $\IDX[i]$.} At time $s = 0$, $\mathcal{S}$ only holds $\tracerand(H_t^0)$ = $(id(\mathcal{D})$, $|\mathcal{D}[1]|, \dots, |\mathcal{D}[n]|$, $h$, $\freqmax$, $\sizemax)$ by which it generates $\IDX[i]$ as follows: first, $\mathcal{S}$ generates a vector of polynomial coefficients by calling $\genVec(\bot, i)$ where $\bot$ denotes an empty list of keywords; then, $\mathcal{S}$ calls $\mathcal{S}_{\FHIPPE}$.\textsf{Encrypt} ($\genVec(\bot, i))$ to compute $\IDX^*[i]$. It should be noted that $\IDX^*[i]$ can be matched by a token generated by $\mathcal{S}_{\FHIPPE}.\textsf{GenToken}$ when inputting $[i^0, i^1, \dots, i^{\sizemax + 1}]$. $\IDX^*$ can be obtained by combining all $\IDX^*[i]$ together. We claim that $\IDX^*$ is indistinguishable from $\IDX$ which will be proved below.

$\bullet$ \textbf{Simulate $\Gamma_k$.} At time $s\leq t$, $\mathcal{S}$ knows $\tracerand(H^s_t)$, by which it simulates $\Gamma_s$ as follows (simulations of $\Gamma_k$, for $k<s$, can be constructed similarly). Initialize an empty multiset $X$. For $i\in[n]$, add $\aprand[{\vec{w}[s]}][i]$ copies of $[i||0||\text{-1}, h(i)]$ to $X$. Then, for $i\in[|h|]$, add $\aprand[{\vec{w}[s]}][n+i]$ copies of $[w_{\text{-1}}||\text{-1}||0,i]$ to $X$. Run lines~\ref{alg1:part4} onward of Alg.~\ref{alg:gentoken} to get $\texttt{PredLabelPairs}$. Transform each predicate in $\texttt{PredLabelPairs}$ into a token by calling $\mathcal{S}_{\FHIPPE}.\textsf{GenToken}$, and the set of (token, label) pairs is $\Gamma_s^*$. We prove next that $\Gamma_s^*$ and $\Gamma_s$ are indistinguishable.

$\bullet$ \textbf{Indistinguishability.} We prove that $\IDX^*$ and $\Gamma_{i}^*$ are indistinguishable from $\IDX$ and $\Gamma_{i}$ by contradiction: if there exists an adversary $\mathcal{A}$ which can distinguish $\IDX^*$ and $\Gamma_{i}^*$ from $\IDX$ and $\Gamma_{i}$, we show that $\mathcal{A}$ breaks SIM-security.
Assume $\mathcal{A}$ can distinguish $\IDX^*$ and $\Gamma_{i}^*$ from $\IDX$ and $\Gamma_{i}$. Let $\{x_{j, 0}\}$ (resp. $\{x_{j, 1} \}$) be the underlying secrets of $\{\IDX[j]\}$ (resp. $\{\IDX^{*}[j]\}$) and $\{y_{j, 0}\}$ (resp. $\{y_{j, 1}\}$) be the underlying secrets of $\Gamma_{i}$ (resp. $\Gamma_{i}^*$). Consider the following two full security games where one is between $\mathcal{A}$ and $\FHIPPE$, denoted by $G_{\mathcal{A}, \FHIPPE}$, and the other one is between $\mathcal{A}$ and $\mathcal{S}_{\FHIPPE}$, denoted by $G_{\mathcal{A}, \mathcal{S}_{\FHIPPE}}$.  
In both games, the adversary $\mathcal{A}$ issues the same queries as follows:
\begin{enumerate}
	\item Ciphertext query. $\mathcal{A}$ queries $(x_{j, 0}, x_{j, 1}), \forall j \in [n]$.
	\item Token query. $\mathcal{A}$ queries $(y_{j, 0}, y_{j, 1}), \forall j \in [|\Gamma_i|]$ where $|\Gamma_i|$ means the number of tokens in $\Gamma_i$.
\end{enumerate}

In $G_{\mathcal{A}, \FHIPPE}$, with probability $1/2$, $b = 0$ and $\FHIPPE$ outputs $I'$ and $\Gamma_{i}'$ (since $\FHIPPE$ is probabilistic) which are indistinguishable from $\IDX$ and $\Gamma_{i}$.  Likewise, in $G_{\mathcal{A}, \mathcal{S}_{\FHIPPE}}$, with probability $1/2$, $b = 1$ and $\mathcal{S}_{\FHIPPE}$ outputs ${\IDX^*}'$ and ${\Gamma_{i}^*}'$ (since $\mathcal{S}_{\FHIPPE}$ is probabilistic) which are indistinguishable from $\IDX^*$ and $\Gamma_{i}^*$.

Since $\mathcal{A}$ can distinguish $\IDX^*$ and $\Gamma_{i}^*$ from $\IDX$ and $\Gamma_{i}$, it must be able to distinguish $I'$ and $\Gamma_{i}'$ from ${\IDX^*}'$ and ${\Gamma_{i}^*}'$, namely $\mathcal{A}$ can distinguish the game with $\FHIPPE$, denoted by $G_{\mathcal{A}, \FHIPPE}$ (the real world) from the game with $\mathcal{S}_{\FHIPPE}$, denoted by $G_{\mathcal{A}, \mathcal{S}_{\FHIPPE}}$ (the ideal world). This contradicts the fact $\FHIPPE$ is SIM-secure. Therefore, $\mathcal{A}$ cannot distinguish $I^*$ and $\Gamma_i^*$ from $I$ and $\Gamma_i$ and thus $\OSSE$ is adaptively semantically secure.

\end{proof}

\section{Differential Privacy Analysis and Discussion}
\label{section:privacy}

In this section, we compare $\OSSE$ and the proposal by Chen et al.~\cite{chen2018differentially}, henceforth termed $\CLRZ$,
\footnote{This acronym takes the first letter of each author of the paper~\cite{chen2018differentially}.} under the differential privacy framework. 
We chose this framework since it is one of the most widely accepted notions of privacy. 
We first propose differential privacy for documents and keywords, which generalize previous SSE-related differential privacy notions~\cite{chen2018differentially}. 
We explain that they imply access and search pattern protection, respectively. 
Then, we prove the differential privacy guarantees of $\OSSE$ and compare them with $\CLRZ$. We show that none of these mechanisms can simultaneously achieve high theoretical privacy guarantees while providing high utility. 
In Sect.~\ref{section:evaluation} we show that these mechanisms offer adequate protection even when differential privacy deems them as weak, since practical attacks widely differ from the strong implicit adversary assumed in the differential privacy framework.

\subsection{Differential Privacy Definitions}

Differential privacy~\cite{dwork2008differential} (DP) is one of the most broadly accepted theoretical notions of privacy, and has been applied to different areas of privacy, including database privacy~\cite{wagh2016root, toledo2016lower, chen2018differentially}. Differential privacy is characterized by a privacy parameter $\epsilon$, and ensures that neighboring inputs to the privacy-preserving algorithm produce any given output with similar probability, and thus it is hard to gain information about the secret input given any observation. Chen et al.~\cite{chen2018differentially} define differential privacy for access pattern leakage. In the definitions below, we generalize their definition and also instantiate it to account for search-pattern privacy. 
Here, we represent an SSE scheme as a function $\mathcal{SE}$ that takes the client input (a dataset $\mathcal{D}$ and a query vector $\vec{w}$) and outputs the leakage (a trace $T$).

\begin{definition}[Differential Privacy for Documents]\label{def:dpdocuments}
An SSE characterized by function $\mathcal{SE}$ provides $\epsilon$-DP for \emph{documents} iff for any keyword list of length $t$, $\vec{w} \in \Delta^t$, for any pair of neighboring databases $\mathcal{D}, \mathcal{D'} \in 2^{2^\Delta }$ (they differ in exactly one position $i$ and exactly one keyword $w$, i.e., $w$ is in either $\mathcal{D}[i]$ or $\mathcal{D}'[i]$ but not both), and any trace $T$, the following holds:
\begin{equation}
    Pr[\mathcal{SE}(\mathcal{D}, \vec{w}) = T] \leq e^{t\cdot \epsilon} Pr[\mathcal{SE}(\mathcal{D'}, \vec{w}) = T]\,.
\end{equation}
\end{definition}

We chose to call this notion ``differential privacy \emph{for documents}'' instead of ``for access patterns'', since the hidden variable in this definition is the database that contains the documents, $\mathcal{D}$. Intuitively, satisfying the above definition guarantees that no one can determine if a document contains a keyword given the trace. This implies that an SSE that provides $\epsilon$-DP for documents also provides \emph{access pattern privacy}.

\begin{definition}[Differential Privacy for Keywords]\label{def:dpkeywords} 
An SSE characterized by function $\mathcal{SE}$ provides $\epsilon$-DP for \emph{keywords} iff for any database $\mathcal{D} \in 2^{2^\Delta}$, for any pair of neighboring keyword lists $\vec{w}, \vec{w}' \in \Delta^{|\vec{w}|} $ ($\vec{w}$ and $\vec{w}'$ differ in only one element, $w[i]\neq w'[i]$), and any trace $T$, the following holds:
    \begin{equation}
    Pr[\mathcal{SE}(\mathcal{D}, \vec{w}) = T] \leq e^{d \cdot \epsilon} Pr[\mathcal{SE}(\mathcal{D}, \vec{w}') = T]\,.
    \end{equation}
    Here, $d$ is the number of different documents between $\mathcal{D}(\vec{w}[i])$ and $\mathcal{D}(\vec{w}'[i])$.
\end{definition}

A scheme that provides DP for keywords guarantees that no one can determine, through observing the trace (allowed leakage), whether a client is searching for one keyword list or the other.
This in turn implies \emph{search pattern privacy}, i.e., it prevents the server from guessing whether or not two queries were performed for the same keyword.

\subsection{Comparison of Differential Privacy Guarantees in SSE}

We compare $\OSSE$ and $\CLRZ$ in terms of differential privacy.
We leave ORAM and PIR schemes out of this comparison since, even though they hide the location of individual documents retrieved, they would require adding differentially private noise to the number of documents retrieved per query; doing this, on top of their already high overhead, would make them impractical.

\begin{theorem}\label{theorem:dp-dpsse}
	$\OSSE$ provides $\epsilon$-differential privacy for documents and keywords, where $\epsilon = \ln\left(\frac{\TPR}{\FPR}\frac{1-\FPR}{1-\TPR}\right)$, with $\TPR=p+(1-p)q$ and $\FPR=q$.
\end{theorem}

The proof is in Appendix~\ref{section:dpproof}.
Table~\ref{table:DPcomp} compares the $\epsilon$ values provided by $\OSSE$ and $\CLRZ$ in terms of their true positive and false positive rates (TPR and FPR).
The table shows that $\CLRZ$ usually provides stronger DP for documents than $\OSSE$ (larger $\epsilon$ indicates weaker privacy guarantees).
On the contrary, $\OSSE$'s guarantee for protecting access patterns is weaker, but it also provides differential privacy for keywords, which implies a certain level of search pattern privacy. 

\begin{table}[t]
\centering
    \begin{tabular}{r | c c}
         & \textbf{DP for documents} & \textbf{DP for keywords}\\ \hline       
        $\CLRZ$ & $\epsilon=\ln\left(\max\left\{\frac{\TPR}{\FPR},\frac{1-\FPR}{1-\TPR}\right\}\right)$ & None ($\epsilon=\infty$) \\
        $\OSSE$ & $\epsilon=\ln\left(\frac{\TPR}{\FPR}\frac{1-\FPR}{1-\TPR}\right)$ & $\epsilon=\ln\left(\frac{\TPR}{\FPR}\frac{1-\FPR}{1-\TPR}\right)$\\ 
        \hline
    \end{tabular}
    \caption{Differential Privacy (DP) guarantees.}\label{table:DPcomp}
\end{table}

Achieving high privacy regimes requires undesirable $\TPR$ and $\FPR$ values for both schemes.
For example, achieving $\epsilon=1$ with a moderately high $\TPR=80\%$ requires $\FPR\approx 45\%$ for $\CLRZ$ and $\FPR\approx 60\%$ for $\OSSE$, which is prohibitive in terms of bandwidth.
However, if the client wants to get a random subset of $\TPR=30\%$ of the documents that match a keyword, they can achieve high privacy ($\epsilon=1$) with only $\FPR=11\%$ (resp.~$13\%$) with $\CLRZ$ (resp.~$\OSSE$).
In this high-privacy low-utility setting, $\CLRZ$ and $\OSSE$ offer similar DP for documents, but $\OSSE$ additionally provides DP for keywords (i.e., search pattern protection).

There are cases where $\TPR$ must be kept close to one and $\FPR$ close to zero for performance issues (e.g., false positives imply a bandwidth increase).
The fact that $\OSSE$ cannot provide a low $\epsilon<1$ in these cases does not mean it is not effective at deterring attackers.
In fact, recent work~\cite{demertzisseal} has shown (for an ORAM-based defense technique) that even a small amount of noise can seriously harm current database and query recovery attacks.
We confirm that this is also true for $\OSSE$ against different attacks in Section~\ref{section:evaluation}.
Indeed, we show that $\TPR=0.9999$ and $\FPR=0.025$, which gives a large $\epsilon\approx 13$ for $\OSSE$, is enough to deter different attacks~\cite{islam2012access, liu2014search, pouliot2016shadow, cash2015leakage}.
$\CLRZ$, however, is significantly more vulnerable to these attacks since it does not provide any search pattern protection.
The explanation for this high protection despite low differential privacy guarantees is that differential privacy assumes a worst-case scenario where the adversary either knows every single entry in the dataset except for one, or knows all of the user queries but one.
Assuming this strong adversary is unrealistic in most cases and, besides that, there exist practical attacks~\cite{islam2012access, liu2014search, pouliot2016shadow, cash2015leakage, blackstonerevisiting} that do not require an adversary this powerful.

\section{Complexity Analysis}
\label{section:complexity}

We study the communication and computation complexity of $\OSSE$.
Note that $\OSSE$ does not require client-side storage after outsourcing the encrypted database and the search index, other than for storing parameters and keys. 
We omit the initialization from the complexity analysis, since it is only done once. Table~\ref{table:notation_complexity} summarizes the notation of this section.

\begin{table}[t]
\centering
    \begin{tabular}{r p{6.5cm}}
        \textbf{Notation} & \textbf{Description}\\ \hline       
        $\freqmax$ & Maximum keyword frequency (maximum number of documents that contain any given keyword). \\
        $\sizemax$ & Maximum no.~of distinct keywords in a document.\\ 
        $\countermax$ & Maximum value of $\counteri{}$ allowed.\\
        $\toksize$ & Size of a query token and label pair (e.g., bytes). \\
        $\docsize$ & Maximum size of a document (e.g., bytes). \\ \hline
    \end{tabular}
    \caption{Notation for the Complexity Analysis}\label{table:notation_complexity}
\end{table}
 
\subsection{Communication Overhead}

The communication complexity refers to the total number of tokens sent from client to server when querying for a keyword $w$, and the number of documents that the server returns to the client as a response. We study the communication overhead, i.e., the  \emph{average} communication cost of $\OSSE$ compared to the average cost of a standard SSE scheme. Let $E_w\doteq\text{E}\{\mathcal{D}(w)\}$ be the expected number of documents that contain keyword $w$. We assume that $q$ is small enough so that $n\cdot q<\freqmax$, and use $|h|=\freqmax$ and $\countermax=3\ln n/\ln\ln\freqmax$.

First, when querying for keyword $w$, the client calls $\genToken$ (Alg.~\ref{alg:gentoken}) to generate tokens. The number of tokens generated in the first loop (lines \ref{alg1:part1}-\ref{alg1:part1end}) is the sum of $|h|\cdot\countermax$ Bernoulli distributions with parameter $p$, while the remaining two loops generate a number of tokens that is the sum of $n$ and $|h|$ Geometric distributions with parameter $1-q$. Therefore, the \emph{expected} number of tokens sent from client to server is
\begin{align*}
 \nmescliser&=|h|\cdot \countermax \cdot p + (n+|h|)\cdot\frac{q}{1-q}\\
    &\approx|h|\cdot \countermax \cdot p + n\cdot q<\freqmax\cdot(\countermax +1)\,,
\end{align*}
where we have used that $n\gg |h|$ and $\frac{q}{1-q}\approx q$ since we are assuming that $q$ is close to zero.

Since the server replies with the documents that are matched \emph{at least once}, the expected number of documents returned is
\begin{equation}
 \nmessercli=E_w(p+q-pq)+(n-E_w) q\leq E_w p +n q<E_w + \freqmax\,.
\end{equation}

The total communication cost of $\OSSE$ is $\nmescliser\cdot\toksize + \nmessercli\cdot\docsize$. A standard SSE just sends a single query message to the server and receives, on average, $E_w$ documents. Therefore, its cost is $\approx E_w\cdot\docsize$. Dividing these two expressions and using some basic algebra, we can express the \emph{communication overhead} of $\OSSE$ as
\begin{equation}
 \overhead{\OSSE}<\frac{\freqmax}{E_w}\left((\countermax+1)\cdot\frac{\toksize}{\docsize}+1\right)+1\,.
\end{equation}

In many cases, $(\countermax+1)<\docsize/\toksize$ (for example, if the database contains images, the encrypted documents can easily be 10 times larger than the token size). In that case, the cost is just $ \overhead{\OSSE}<2\cdot \freqmax/E_w+1$.

This cost depends on $E_w$, which in turn depends on the query and keyword distribution. We perform the analysis for three possible distributions: uniform, Zipfian (typical in natural language~\cite{zipf1932selected}), and the worst-case distribution for $\OSSE$ that considers that the queried keywords appear in a single document. We show how to compute $E_w$ for each of these distributions in Appendix~\ref{app:distributions}. We get
\begin{equation}
 \overhead{\OSSE}<\begin{cases}
          3\,,\qquad&\text{uniform,}\\
          1.36+0.61\cdot \log |\Delta|\,,&\text{Zipf,}\\
          1+2\cdot \freqmax\,,&\text{worst-case.}\\
         \end{cases}
\end{equation}
The cost is increased by a factor no bigger than $\countermax$+1 for datasets where each document size is comparable to the size of each keyword.

\subsection{Computational Complexity}\label{section:comp_complexity}

We disregard the initialization cost (generating the index) since it is only done once at the beginning of the protocol. We measure the computational complexity of a query as the number of evaluations that the server needs to compute when searching for a keyword $w$. Every time the server receives a token, it calls the function $\FHIPPE$.$\textsf{Query}$ on average $n/|h|$ times (since this is the average number of documents with a given label). Therefore, the number of calls to $\FHIPPE$.$\textsf{Query}$ is
\begin{equation}
 \computation{\OSSE}=\nmescliser\cdot \frac{n}{|h|}<n\cdot(\countermax+1)\,.
\end{equation}

\subsection{Faster $\OSSE$}
\label{section:dpsse_fast}

$\OSSE$ uses a hash function $h$ to assign a label to each document. This requires a certain $\countermax$ value to ensure that no token can match more than one document. We explain a variation of the algorithm that assigns labels to documents that achieves a lower $\countermax$ and therefore reduces the communication and computational cost of $\OSSE$. In this variation, there are two publicly known hash functions $h_1$ and $h_2$. When building the search index, the client generates roots $(w||l||\counteri{w,l})$ by choosing $l$ between $h_1(id(D))$ and $h_2(id(D))$ as the one that minimizes the counter. 
This way, the value of $\countermax$ is only $O(\ln \ln \freqmax)$~\cite{azar1999balanced}. 
This reduction of $\countermax$ can provide significant computational and communication advantages, but at the cost of a considerable reduction of the differential privacy guarantees (we omit this analysis here for space constraints). 
This modification of $\OSSE$ is still highly useful when differential privacy protection is not required (e.g., against known attacks).

\subsection{Comparison with Other Schemes}

The only searchable ORAM-based SSE scheme is TWORAM~\cite{garg2016tworam}, which has a communication overhead of $O(\log n \log\log n)$ (at least four rounds) and client-side storage of $O(\log^2 n)$. As explained above, under some realistic assumptions on the query and keyword distribution, $\OSSE$ requires less communication cost. Also, $\OSSE$ does not require client-side storage except for storing parameters like $p$ and $q$, and only uses one communication round.

Even though Fully Homomorphic Encryption (FHE) schemes are competitive performance-wise~\cite{akavia2018secure, akavia2019setup}, they can only return a fixed, constant number of results on each query due to the fixed length of each circuit.
Adapting FHE schemes to variable length results requires multiple communication rounds, whereas our scheme requires a single round.

The only practical privacy-preserving SSE scheme we are aware of is $\CLRZ$~\cite{chen2018differentially}.
In $\CLRZ$, the client sends a single query token, and receives on average the same amount of documents from the server as in $\OSSE$.
\footnote{Chen et al.~\cite{chen2018differentially} alleviate the false negatives by introducing document redundancy, but we do not consider this here for a fair comparison. Our proposal is also compatible with this document redundancy approach.}
Thus, the communication overhead of $\CLRZ$ is slightly smaller than that of $\OSSE$.
The computational cost of $\CLRZ$ depends on the underlying SSE that it implements; when combined with an inexpensive SSE, $\CLRZ$ can achieve less overhead than $\OSSE$.

\section{Evaluation}
\label{section:evaluation}

We evaluate the effectiveness and efficiency of $\OSSE$. First, we evaluate $\OSSE$ and $\CLRZ$~\cite{chen2018differentially} against different query recovery attacks:
\begin{enumerate}
	\item Liu et al.'s \emph{frequency attack}~\cite{liu2014search}, that recovers queries using query frequency information and search patterns.
	\item The \emph{IKK attack}~\cite{islam2012access}, that recovers the queries using keyword co-occurrence data, ground-truth information of a subset of the client queries, and access and search patterns.
	Chen et al.~use this attack to evaluate $\CLRZ$~\cite{chen2018differentially}.
	\item Cash et al.'s \emph{count attack}~\cite{cash2015leakage}, that builds sets of candidate keywords for each query based on the query response volume, and then refines these sets by using co-occurrence information until only one feasible assignment remains.
	\item The \emph{graph matching} attack by Pouliot et al.~\cite{pouliot2016shadow}, that uses keyword co-ocurrence information and does not require any ground-truth information about the queries or data set to be effective.
\end{enumerate}
Blackstone et al.~\cite{blackstonerevisiting} recently proposed a query recovery attack (\emph{subgraph attack}) that uses access pattern leakage and partial database knowledge.
The attack tries to match queries to keywords by identifying patterns from the partial database knowledge in the observed access patterns.
Since $\CLRZ$ and $\OSSE$ randomize the access patterns, adapting the attack against these defenses is not trivial, and therefore we do not include it in our evaluation.

We run $\CLRZ$ without using its document redundancy technique~\cite{chen2018differentially}: this technique aims at reducing false negatives and is trivially compatible with $\OSSE$; we decide not to use it for simplicity and because we believe it does not affect performance results.
For a fair evaluation, we \emph{adapt} the attacks above to perform well against $\CLRZ$ and $\OSSE$. 
We explain the details in Appendix~\ref{app:adapt}.\footnote{Our evaluation code is available at \href{https://github.com/simon-oya/NDSS21-osse-evaluation}{https://github.com/simon-oya/NDSS21-osse-evaluation}}
  
Second, we report the running time of a prototypical implementation of $\OSSE$ in Python on an Intel(R) E7-8870 160-core Ubuntu 16.04 machine clocked at 2.40 GHz with 2 TB of system memory.\footnote{Our prototype implementation is available at \href{https://github.com/z6shang/OSSE}{https://github.com/z6shang/OSSE}}

We use the Enron email dataset in our experiments,\footnote{\url{https://www.cs.cmu.edu/~enron/}} which contains 30\,109 emails.
We follow the keyword extraction process in related work~\cite{islam2012access}, then we ignore keywords that do not appear in the English dictionary and remove stopwords; we use the following $|\Delta|$ most common keywords remaining as keyword universe ($|\Delta|$ varies between experiments).

\subsection{Query Frequency Attack}
\label{sec:freq}

We implement and run the first frequency attack by Liu et al.~\cite{liu2014search}. 
We take $|\Delta|=250$ keywords and download their query frequency over 50 consecutive weeks\footnote{The last week we take is the second week of April, 2020.} from Google Trends.\footnote{\url{trends.google.com}}
We store this frequency information in a $50\times 250$ matrix $F$ where $F[i,j]$ is the probability that the client queries for the $j$th keyword in the $i$th week.
The client generates $N_q\in\{100,300,1\,000\}$ queries per week following the distributions in $F$.
We evaluate the case where the client uses no defense, when they use $\CLRZ$, and when they use $\OSSE$.
The server observes the (possibly obfuscated) access patterns and groups them into $N_c$ clusters using a $k$-means algorithm (this is trivial for $\CLRZ$).
We give the adversary the \emph{true} number of distinct keywords queried by the client to use as $N_c$.
This is a worst-case against $\OSSE$.
The adversary labels all queries in cluster $k\in[N_c]$ as the $j$th keyword, where $j$ is the column in $F$ that is closest, in Euclidean distance, to the frequency trend of the cluster over the 50 weeks.

Figure~\ref{fig:liu} shows the attack accuracy (number of correctly identified queries divided over the total number of queries) averaged over 20 runs. 
The shaded areas are the $95\%$ confidence intervals for the mean.  
We vary the $\FPR$ and query number $N_q$, and fix $\TPR=0.9999$.
Since clustering access patterns (i.e., finding the search pattern) is trivial for $\CLRZ$, the accuracy against this defense is independent of $\FPR$.
$\OSSE$ achieves higher protection against the attack since it obfuscates the access patterns of each query independently.

\begin{figure}[t]
\centering
\includegraphics[width=.8\linewidth]{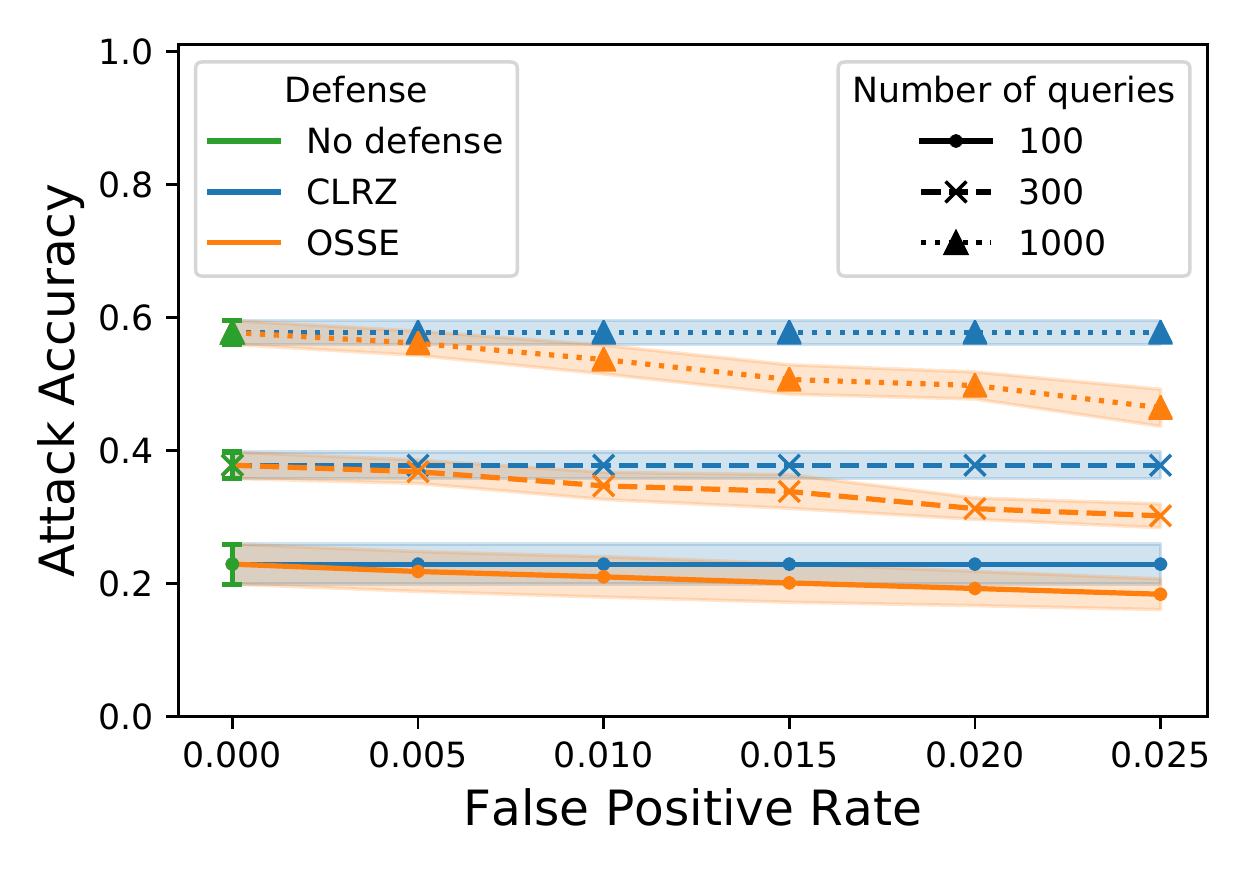}
\caption{Accuracy of Query Frequency Attack}
\label{fig:liu}
\end{figure}

\subsection{IKK Attack}

We run the IKK attack~\cite{islam2012access} using a keyword universe size $|\Delta|=500$. 
This attack relies on the keyword co-occurrence, i.e., a $|\Delta|\times|\Delta|$ matrix $M$ where $M[i,j]$ is the percentage of documents that have both the $i$th and $j$th keyword.
The attack computes this matrix from a training set, and compares it with a co-occurrence matrix computed from the observed access patterns, $M'$.
Unlike in the original IKK attack setting, where the queries are unique, we allow repeated queries and modify the attack accordingly.
We also adapt the attack to perform better against the obfuscated access patterns generated by $\OSSE$ and $\CLRZ$ (see Appendix~\ref{app:adapt}).
We call this improved version IKK$^*$.
In each experiment, the client generates $N_q=400$ queries (using a Zipfian distribution), we give each one with probability $0.15$ to the attacker. 
We use the same Enron data for training and testing~\cite{islam2012access}.\footnote{We use the following attack-specific parameters: initial temperature $200$, cooling rate $0.9999$, and reject threshold $1\,500$.}

Figure~\ref{fig:ikk} shows the query recovery rate, averaged over 20 runs, for $\TPR=0.9999$ and different $\FPR$ values.
We plot the average baseline $15\%$ ground-truth queries known to the attacker as reference.
We see that $\IKKS$ improves over $\IKK$, as expected, and that $\OSSE$ reduces the attack accuracy almost by half compared to $\CLRZ$.
This is due to two facts: first, $\OSSE$ hides the search patterns and makes it hard for the adversary to use the known queries ($15\%$) towards identifying other queries.
Second, since each access pattern is generated independently in $\OSSE$, the search space for $\IKK$ against this attack is of size $|\Delta|^{200}$, which is much larger than $|\Delta|!/(|\Delta|-200)!$, the size of the search space in $\CLRZ$.

\subsection{Count Attack}

The count attack~\cite{cash2015leakage} uses volume information to build a set of candidate keywords for each query and refines them with co-occurrence information until only one possible matching remains.
We optimize the attack against $\OSSE$ and $\CLRZ$ by modifying a generalization of the attack found in the original paper~\cite{cash2015leakage} (see Appendix~\ref{app:adapt}).
We give the adversary the plaintext database to build the auxiliary volume and co-occurrence information.
However, the adversary does not see any ground-truth query.
We run the experiments for the same parameters as IKK ($|\Delta|=500$, $N_q=400$, Zipfian distribution for queries).
When the algorithm fails to find any plausible matching, we set the accuracy to $1/|\Delta|$ (random guessing).

Figure~\ref{fig:count} shows the accuracy of the count attack (20 runs) as well as the frequency of failure (inconsistency rate) against $\OSSE$.
The count attack achieves perfect accuracy when no defense is applied, but its accuracy already decreases to $0.65$ against $\CLRZ$ with $\TPR=0.9999$ and $\FPR=0$.
The accuracy remains stable against $\CLRZ$ as $\FPR$ increases, but it sharply decreases against $\OSSE$, mostly due to inconsistencies.
These inconsistencies stem from the fact that the heuristics of the count attack require that the observed volume and co-occurrences of queries fall inside certain expected intervals.
$\OSSE$ generates freshly random access patterns, so the probability that at least one observed volume falls outside an interval is high, which explains the high inconsistency rate.

\begin{figure*}[t]
\begin{minipage}{0.32\linewidth}
	\centering
	\includegraphics[width=\linewidth]{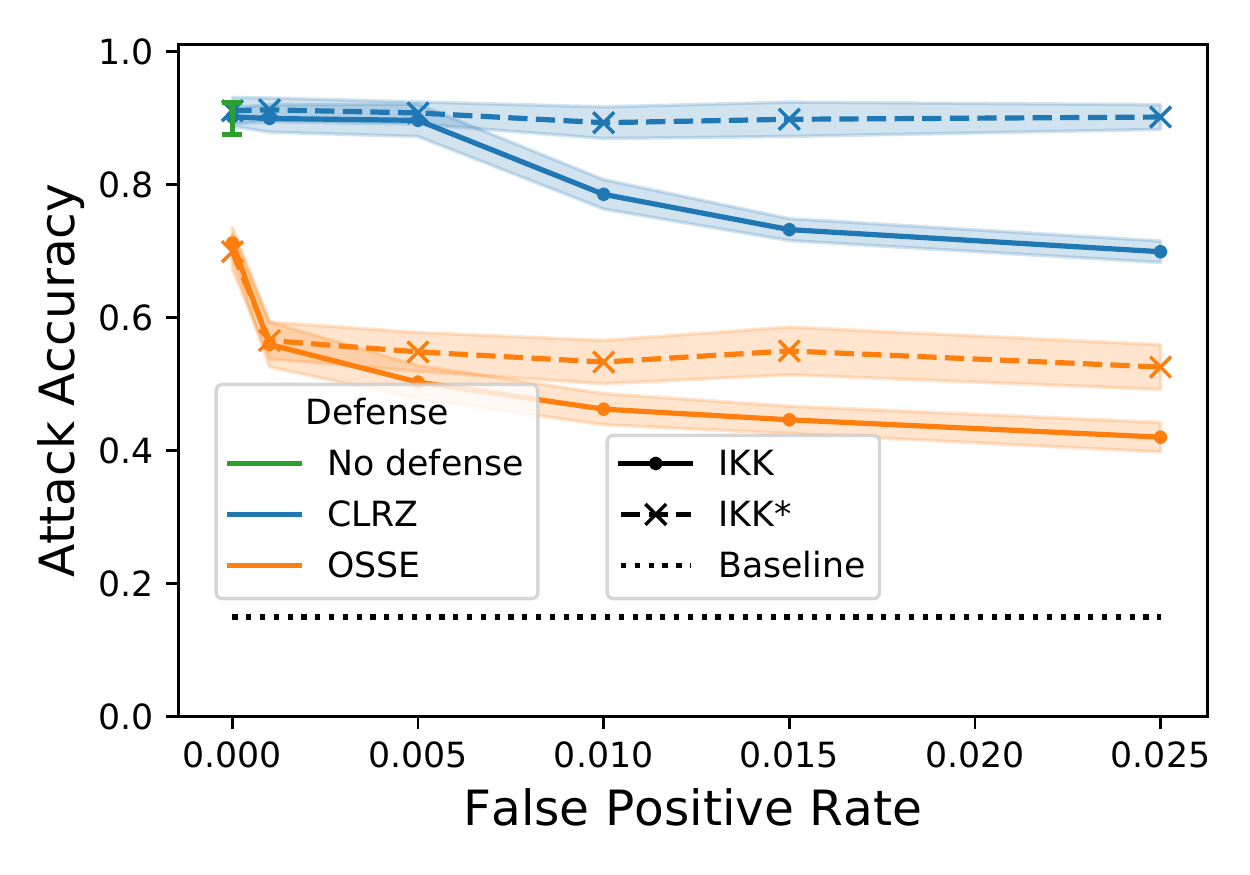}\\
	\caption{Accuracy of $\IKK$ Attack}% (known dataset, $15\%$ ground-truth queries)}
	\label{fig:ikk}
\end{minipage} \hfill
\begin{minipage}{0.32\linewidth}
	\centering
	\includegraphics[width=\linewidth]{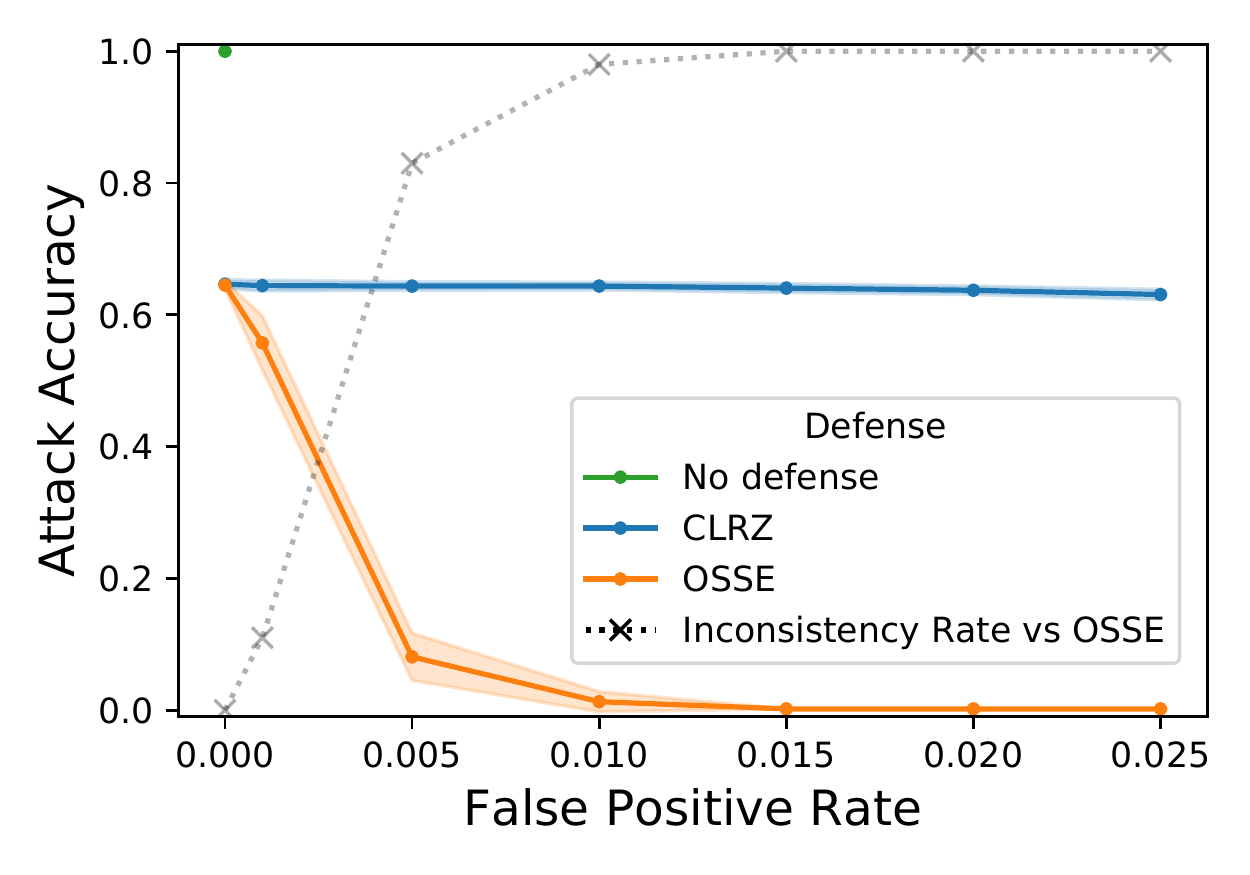}\\
	\caption{Accuracy of Count Attack}% (known dataset)}
	\label{fig:count}
\end{minipage} \hfill
\begin{minipage}{0.32\linewidth}
	\centering
	\includegraphics[width=\linewidth]{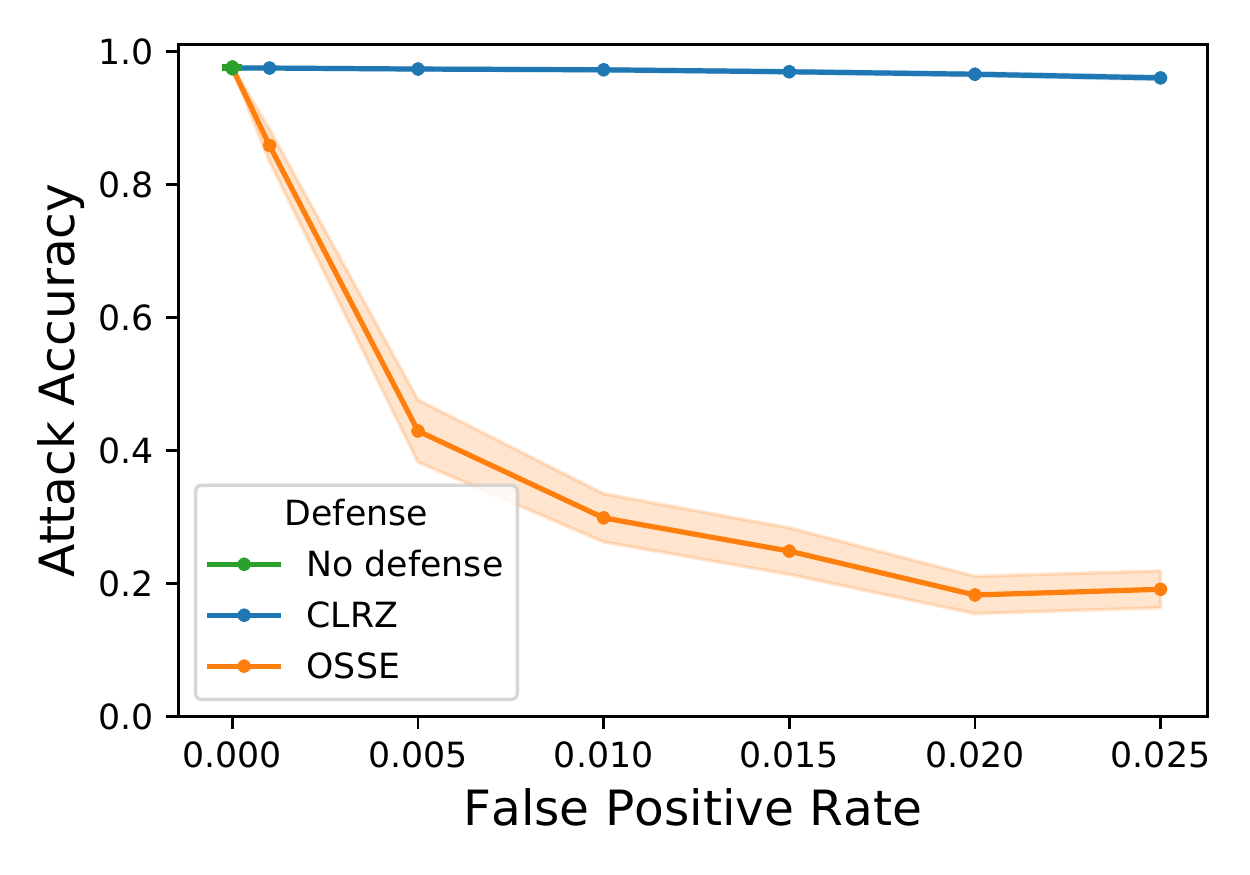}\\
	\caption{Accuracy of Graph Matching}% (train/test dataset split)}
	\label{fig:graphm}
\end{minipage}
\end{figure*}

\subsection{Graph Matching Attack}

We implement the graph matching attack~\cite{pouliot2016shadow} using the PATH algorithm~\cite{zaslavskiy2008path}. 
We randomly split Enron emails evenly into a training and testing set, and give the adversary the training set only (we redo the split in each run).
The adversary does not have any ground truth information about the queries.
This is a more realistic scenario than considered by $\IKK$ and the count attack. 
We take $|\Delta|=250$ keywords instead of 500, since this attack is computationally demanding.
This attack requires a large number of queries to work properly, since it has imperfect information about the dataset.
We use $N_q=2\,000$ queries (Zipfian distribution).\footnote{The objective function of the graph matching attack has an hyperparameter $\alpha\in[0,1]$. We used $\alpha=0$ since it yielded the highest accuracy.}

Figure~\ref{fig:graphm} shows the query recovery rate of the graph matching attack, averaged over 100 runs, for $\TPR=0.9999$ and different $\FPR$ values.
The attack achieves an accuracy of $\approx 99\%$ against $\CLRZ$, which is surprisingly large considering that the adversary has imperfect information. 
The performance against $\OSSE$ drops below $50\%$ with only $0.5\%$ false positives and keeps decreasing as $\FPR$ grows.

\subsection{Running Time}
We measure the average running time of the $\OSSE$ functions $\textsf{BuildIndex}$, $\textsf{Trapdoor}$, and $\textsf{Search}$ on Enron dataset.
We use the function-hiding inner product encryption scheme by Kim et al.~\cite{kim2018function} with the $\texttt{MNT159}$ pairing curve.
We limit the number of keywords in each document to $\sizemax=300$ by splitting large documents into smaller ones.
By default, $97\%$ of the documents have no more than 300 keywords.
After splitting, the number of documents increases to $n=30\,562$ and $\freqmax\approx 2\,000$.
With the single hashing method, we get $\countermax=7$, while the dual hashing method explained in Sect.~\ref{section:dpsse_fast} yields $\countermax=3$.
We evaluate the running time of the latter. We use $p=0.9999$ and $q=0.01$.

Table~\ref{table:run-time-parallel} summarizes the running times. 
Chen et al.~\cite{chen2018differentially} report running times for $\CLRZ$ in Enron dataset around 100 seconds for obfuscating the search index, and less than 200 milliseconds per query.
Even though the running times of $\OSSE$ in our experiment (e.g., half an hour for a search) can be reduced by switching from Python to a more efficient language and parallelizing the queries, our scheme is substantially slower than $\CLRZ$.
In exchange, $\OSSE$ provides search pattern obfuscation, a rare property that we have shown provides significant protection against different query recovery attacks.

\begin{table}[t]
\centering
    \begin{tabular}{cccc}
        \hline
        \texttt{\# cores} & \textsf{BuildIndex} (min) & \textsf{Trapdoor} (s) & \textsf{Search} (min)\\
         \hline 
         4 & 272.5  & 580.7  & 1099.1  \\
         8 & 136.3  & 290.5  & 549.6  \\
         16 & 68.2  & 145.3  & 274.8  \\
         32 & 34.1  & 72.8  & 137.4  \\
         64 & 17.1  & 36.4  & 68.7  \\
         128 & 8.5  & 18.2  & 34.4  \\
         160 & 6.9  & 14.7  & 27.5  \\
         \hline 
    \end{tabular}
    \caption{Running Times}\label{table:run-time-parallel}
\end{table}

\subsection{Summary of Results}

Our experiments reveal that $\OSSE$ achieves higher privacy protection than $\CLRZ$ against a variety of query recovery attacks.
This includes a search pattern-based attack~\cite{liu2014search} and access pattern-based attacks with ground-truth information about queries and dataset (IKK~\cite{islam2012access}), ground-truth information about the dataset (count~\cite{cash2015leakage}), and no ground-truth information (graph matching~\cite{pouliot2016shadow}).
The advantage of $\OSSE$ over $\CLRZ$ comes from the fact that $\OSSE$ generates each access pattern independently at random, which hides the search pattern.
Interestingly, our evaluation shows that this undermines attacks that rely on access pattern leakage.
This is because these attacks~\cite{islam2012access, cash2015leakage, pouliot2016shadow} formulate the query recovery problem by implicitly assuming the adversary can distinguish between queries for distinct keywords (i.e., search pattern leakage).
When this is not true, the adversary needs to increase their search space (IKK) or perform clustering before running the attack (frequency attack, count attack, and graph matching).
This causes an extra source of error for the adversary.

All these privacy benefits come at the cost of a high running time.
Our scheme is therefore adequate when the data owner values privacy over running time, and when the system is bandwidth-constrained and cannot afford to use multi-round bandwidth-demanding schemes such as ORAM-based ones~\cite{garg2016tworam}.

\section{Conclusions}
\label{section:conclusions}

Searchable Symmetric Encryption (SSE) allows a data owner to outsource its data to a cloud server while maintaining the ability to search over it.
Existing SSE schemes either prevent leakage but are prohibitive in terms of their cost, or are efficient but extremely vulnerable to query and database identification attacks.
In 2018, Chen et al.~proposed a middle-ground solution that is cost efficient and partially protects which documents match each query, i.e., the access pattern. However, this scheme leaks whether or not the same keyword is being searched multiple times, i.e., the search pattern.

In this work, we propose $\OSSE$, an adaptively semantically secure SSE scheme that obfuscates both access and search patterns by generating each query randomly and independently of previous queries.
We prove that the communication overhead of $\OSSE$ can be a small constant when the keyword distribution in the dataset is uniform, and $O(\log |\Delta|)$, where $|\Delta|$ is the keyword universe, when the keyword and query distributions are Zipfian.
Although it has a large computation complexity, $\OSSE$ is easily parallelizable, and performs a search in a single communication round.
Our evaluation shows that $\OSSE$ is highly effective against current query identification attacks while providing high utility, and demonstrates the importance of hiding search patterns in privacy-preserving SSE schemes. 

\section*{Acknowledgment}
We gratefully acknowledge the support of NSERC for grants RGPIN-05849, CRDPJ-531191, IRC-537591 and the Royal Bank of Canada for funding this research.
Andreas Peter was supported by the Netherlands Organisation for Scientific Research (Nederlandse Organisatie voor Wetenschappelijk Onderzoek, NWO) in the context of the SHARE project.
This work benefited from the use of the CrySP RIPPLE Facility at the University of Waterloo.

\bibliographystyle{IEEEtranS}
\bibliography{dp-sse-biblio}

\appendix

\subsection{Differential Privacy Guarantees of $\OSSE$}
\label{section:dpproof}

Before deriving the proof for Theorem~\ref{theorem:dp-dpsse}, we introduce some results on ratios of probability distributions.

\subsubsection{Upper Bounds on Probability Ratios}
\label{app:bounds}

Let $G\sim\Geometric{}{1-q}$ be a random variable that follows a geometric distribution defined over $\{0,1,2,\dots\}$, let $A\sim\Bernoulli{p}$ be a Bernoulli random variable, and let $B_n\sim\Binomial{n}{p}$ be a binomial random variable. We prove the following results.

\begin{lemma}
 For any non-negative integer $\alpha\in\mathbb{Z}_+$,
 \begin{align}
  \frac{\Pr(G+A=\alpha)}{\Pr(G=\alpha)}&\leq\frac{p+q(1-p)}{q}\,, \label{eq:GAvsG}\\
 \frac{\Pr(G=\alpha)}{\Pr(G+A=\alpha)}&\leq\frac{1}{1-p}\,. \label{eq:GvsGA}
 \end{align}
\end{lemma}

\begin{proof}
 The ratio in \eqref{eq:GAvsG} takes the following values depending on $\alpha$:
 \begin{equation*} \label{eq:berngeobound}
\Fitpage{
 \frac{\Pr(G+A=\alpha)}{\Pr(G=\alpha)}= \begin{cases}
        \frac{(1-p)(1-q)}{1-q} = 1 - p, &\text{if } \alpha=0;\\
        \frac{pq^{\alpha-1}(1-q) + (1-p)q^{\alpha}(1-q)}{q^{\alpha}(1-q)} = \frac{p + q(1-p)}{q}, &\text{if } \alpha>0\,.
    \end{cases}
		}
\end{equation*}
The bounds in \eqref{eq:GvsGA} and \eqref{eq:GAvsG} follow from the fact that $1-p<\frac{p + q(1-p)}{q}$.
\end{proof}

\begin{lemma}
 For any non-negative integers $\alpha, n\in\mathbb{Z}_+$,
 \begin{align}
  \frac{\Pr(G+B_{n+1}=\alpha)}{\Pr(G+B_{n}=\alpha)}&\leq\frac{p+q(1-p)}{q}\,, \label{eq:GBvsG}\\
  \frac{\Pr(G+B_{n}=\alpha)}{\Pr(G+B_{n+1}=\alpha)}&\leq\frac{1}{1-p}\,. \label{eq:GvsGB}
 \end{align}
\end{lemma}

\begin{proof}
 First, we expand the ratio:
 \begin{align}
  &\frac{\Pr(G+B_{n+1}=\alpha)}{\Pr(G+B_{n}=\alpha)}\\
  =& \frac{\sum_{\beta=0}^{\min(\alpha,n+1)} \Pr(B_{n+1}=\beta)\cdot \Pr(G=\alpha-\beta)}{\sum_{\beta=0}^{\min(\alpha,n)} \Pr(B_{n}=\beta)\cdot \Pr(G=\alpha-\beta)}\\
  =&\frac{\sum_{\beta=0}^{\min(\alpha,n+1)} {n+1\choose \beta} p^\beta(1-p)^{n-\beta+1} q^{\alpha-\beta}(1-q)}{\sum_{\beta=0}^{\min(\alpha,n)} {n\choose \beta} p^\beta(1-p)^{n-\beta} q^{\alpha-\beta}(1-q)}\,.\label{eq:appstep}
 \end{align}

 We use ${n+1\choose k}={n\choose k}+{n\choose k-1}$ with the convention that ${n\choose k}=0$ when $k<0$ or $k>n$, and we perform a change of variable $\beta'\doteq \beta-1$ in the numerator. 
Then, \eqref{eq:appstep} equals: 
 \begin{align}
  &=1-p+\frac{p}{q}\cdot \frac{\sum_{\beta'=0}^{\min(\alpha-1,n)} {n\choose \beta'} p^{\beta'}(1-p)^{n-\beta'} q^{\alpha-\beta'}(1-q)}{\sum_{\beta=0}^{\min(\alpha,n)} {n\choose \beta} p^\beta(1-p)^{n-\beta} q^{\alpha-\beta}(1-q)}\\
  &\begin{cases}
     =1-p\qquad&\text{ if }\alpha=0\,,\\
     =1-p+p/q&\text{ if }\alpha> n\,,\\
     \in (1-p,1-p+p/q)&\text{ otherwise.}
    \end{cases}
 \end{align}
 Therefore, the ratio is smaller than $1-p+p/q$, which proves \eqref{eq:GBvsG}, and its inverse is smaller than $1/(1-p)$, which proves \eqref{eq:GvsGB}.
\end{proof}

\begin{lemma}
 For any non-negative integers $\alpha, n, m\in\mathbb{Z}_+$,
\begin{align}
 \frac{\Pr(G+B_{n+m}=\alpha)}{\Pr(G+B_{n}=\alpha)}&\leq \left(\frac{p+q(1-p)}{q}\right)^m\,, \label{eq:GBvsG2}\\
 \frac{\Pr(G+B_{n}=\alpha)}{\Pr(G+B_{n+m}=\alpha)}&\leq\left(\frac{1}{1-p}\right)^m\,. \label{eq:GvsGB2}
\end{align}
\end{lemma}

\begin{proof}
 We can expand the left ratio in \eqref{eq:GBvsG2} as
 \begin{equation}
  \frac{\Pr(G+B_{n+m}=\alpha)}{\Pr(G+B_{n}=\alpha)}=\prod_{k=1}^m \frac{\Pr(G+B_{n+k}=\alpha)}{\Pr(G+B_{n+k-1}=\alpha)}\,.
 \end{equation}
 Then, by applying the bound in \eqref{eq:GBvsG} to each of the $m$ terms of this product, we reach \eqref{eq:GBvsG2}. Likewise, we can prove \eqref{eq:GvsGB2} by applying \eqref{eq:GvsGB} $m$ times.
\end{proof}

\subsubsection{Differential Privacy for Documents}

We derive the $\epsilon$ guarantee that $\OSSE$ provides according to Definition~\ref{def:dpdocuments}. 
We use $\mathcal{M}$ to denote the randomized mechanism that takes the outsourced dataset $\mathcal{D}$ and a \emph{single} query for keyword $w$ and outputs the obfuscated access pattern for that $w$, i.e,. $\aprand[w]$. Recall that $\aprand[w]$ is a $n+|h|$-sized random vector characterized by \eqref{eq:dist1} and \eqref{eq:dist2}.

Let $\mathcal{D}$ and $\mathcal{D}'$ be two adjacent datasets that are identical except for their $k$th document ($\mathcal{D}[k]$ and $\mathcal{D}'[k]$, respectively). These documents differ in a single keyword. Let $w_*$ be the keyword that is only in $\mathcal{D}[k]$, and $w'_*$ the keyword that is only in $\mathcal{D}'[k]$. Let $\aprand[w]$ be the random output of $\mathcal{M}(\mathcal{D},w)$, and $\aprand[w]'$ be the random output of $\mathcal{M}(\mathcal{D}',w)$. Let $\apsamp[]$ be a particular observed obfuscated access pattern. 
Then, we want to find an upper bound for the ratio
\begin{equation}
\Fitpage{
 \frac{\Pr(\mathcal{M}(\mathcal{D},w)=\apsamp)}{\Pr(\mathcal{M}(\mathcal{D}',w)=\apsamp)}=\frac{\Pr(\aprand[w]=\apsamp[])}{\Pr(\aprand[w]'=\apsamp[])}=\prod\limits_{i=1}^{n+|h|}\frac{\Pr(\aprand[w][i]=\apsamp[][i])}{\Pr(\aprand[w]'[i]=\apsamp[][i])}
}
\end{equation}
If $w\notin\{w_*,w'_*\}$, then this ratio is equal to one. Now consider the case where $w=w_*$. In that case, the distribution of $\aprand[w][i]$ and $\aprand[w]'[i]$ is the same, except for $i=k$ (since $w$ is in $\mathcal{D}[k]$ but not in $\mathcal{D}'[k]$) and $i=l_k$, where $l_k=h(k)$ is the label of the $k$th document. Therefore, we have
\begin{equation}
\Fitpage{
 \frac{\Pr(\mathcal{M}(\mathcal{D},w_*)=\apsamp)}{\Pr(\mathcal{M}(\mathcal{D}',w_*)=\apsamp)}=\frac{\Pr(\aprand[w_*][k]=\apsamp[][k])}{\Pr(\aprand[w_*]'[k]=\apsamp[][k])}\cdot \frac{\Pr(\aprand[w_*][l_k]=\apsamp[][l_k])}{\Pr(\aprand[w_*]'[l_k]=\apsamp[][l_k])}
}
\end{equation}

Since $\aprand[w_*][k]\sim\Bernoulli{p}+\Geometric{}{1-q}$ and $\aprand[w_*]'[k]\sim\Geometric{}{1-q}$, we can bound the left coefficient using \eqref{eq:GAvsG}. Likewise, since $\aprand[w_*][l_k]\sim\Binomial{g}{p}+\Geometric{}{1-q}$ and $\aprand[w_*]'[l_k]\sim\Binomial{g+1}{p}+\Geometric{}{1-q}$ for some positive integer constant $g$, we can bound the right coefficient using \eqref{eq:GvsGB}. 

Therefore, using \eqref{eq:GAvsG} and  \eqref{eq:GvsGB} we obtain
\begin{equation} 
 \frac{\Pr(\mathcal{M}(\mathcal{D},w_*)=\apsamp)}{\Pr(\mathcal{M}(\mathcal{D}',w_*)=\apsamp)}\leq \frac{p + (1-p)q}{q}\cdot\frac{1}{1-p}=1+\frac{p}{q(1-p)}\,.
\end{equation}
We can follow the same procedure to prove the same bound when the client queries for $w_*'$.

Finally, if we consider a sequence of $t$ queries, in the worst case all of them are for either $w_*$ or $w'_*$, so the differential privacy guarantee according to Definition~\ref{def:dpdocuments} is
\begin{equation} \label{eq:epsdocs}
 \epsilon=\ln\left(1+\frac{p}{q(1-p)}\right)\,.
\end{equation}

Using $\TPR=p+q(1-p)$ and $\FPR=q$, we obtain the value in Theorem~\ref{theorem:dp-dpsse}.

\subsubsection{Differential Privacy for Keywords}

Consider a database $\mathcal{D}$, and a pair of neighbouring keyword lists $\Vec{w}, \Vec{w}' \in \Delta^{|\Vec{w}|}$ that are identical except for their $k$th element, that is $w$ in $\Vec{w}$ and $w'$ in $\Vec{w}'$. Let $\aprand[w]$ be the output of $\mathcal{M}(\mathcal{D},w)$ and $\aprand[w']$ be the output of $\mathcal{M}(\mathcal{D},w')$. 
Then, we want to find an upper bound for
\begin{equation}
\Fitpage{
 \frac{\Pr(\mathcal{M}(\mathcal{D},w)=\apsamp)}{\Pr(\mathcal{M}(\mathcal{D},w')=\apsamp)}=\frac{\Pr(\aprand[w]=\apsamp[])}{\Pr(\aprand[w']=\apsamp[])}=\prod\limits_{i=1}^{n+|h|}\frac{\Pr(\aprand[w][i]=\apsamp[][i])}{\Pr(\aprand[w'][i]=\apsamp[][i])}\,.
}
\end{equation}
For $i\in[n]$, the distribution of $\aprand[w][i]$ and $\aprand[w'][i]$ will be different if $\mathcal{D}[i]$ contains only one of $\{w,w'\}$. Let $D_{0,1}$ be the indices of documents that do not contain $w$ but contain $w'$. For $i\in D_{0,1}$, $\aprand[w][i]\sim\Geometric{}{1-q}$ and $\aprand[w'][i]\sim\Bernoulli{p}+\Geometric{}{1-q}$. The ratio between the probabilities of these distributions can be bounded using \eqref{eq:GvsGA}. Likewise, define $D_{1,0}$ (indices of documents that contain $w$ but do not contain $w'$), $D_{1,1}$ (contain both), and $D_{0,0}$ (contain neither). Using \eqref{eq:GAvsG} and \eqref{eq:GvsGA} we can get:
\begin{equation}  
 \frac{\Pr(\aprand[w][i]=\apsamp[][i])}{\Pr(\aprand[w'][i]=\apsamp[][i])}\leq
 \begin{cases}
    1 &\text{if }i\in D_{0,0}\\
    \frac{1}{1-p} &\text{if }i\in D_{0,1}\\ 
    \frac{p+q(1-p)}{p} &\text{if }i\in D_{1,0}\\
    1 &\text{if }i\in D_{1,1}\\    
 \end{cases}
\end{equation}
Therefore,
\begin{equation} \label{eq:dpkw1}
\Fitpage{
 \prod\limits_{i=1}^n \frac{\Pr(\aprand[w][i]=\apsamp[][i])}{\Pr(\aprand[w'][i]=\apsamp[][i])} \leq \left(\frac{1}{1-p}\right)^{|D_{0,1}|} \left(\frac{p+q(1-p)}{p}\right)^{|D_{1,0}|}\,.
}
\end{equation}

Now, we focus on the variables $\aprand[w][i+n]$ for $i\in[|h|]$. Let $D_{0,1}^i$ be the indices of the documents in $D_{0,1}$ whose label is $i$ (same for the other sets $D_{0,0}$, $D_{1,0}$, and $D_{1,1}$, for $i\in[|h|]$). Then, we can write
\begin{align}
 \aprand[w][i+n] &\sim \Binomial{g_{i+n}+|D_{0,1}^i|}{p} + \Geometric{}{1-q}\,;\\
 \aprand[w'][i+n] &\sim \Binomial{g_{i+n}+|D_{1,0}^i|}{p} + \Geometric{}{1-q}\,,
\end{align}
for some constants $g_{1+n}, g_{2+n}, \dots, g_{|h|+n}$. If $|D_{0,1}^i|<|D_{1,0}^i|$, we can define a different set of constants $g'_{i+n}\doteq g_{i+n}+|D_{0,1}^i|$ and write
\begin{align}
 \aprand[w][i+n] &\sim \Binomial{g'_{i+n}}{p} + \Geometric{}{1-q}\,;\\
 \aprand[w'][i+n] &\sim \Binomial{g'_{i+n}+|D_{1,0}^i|-|D_{0,1}^i|}{p} + \Geometric{}{1-q}\,,
\end{align}
Using \eqref{eq:GvsGB2}, we can write
\begin{equation}  \label{eq:dpkw2}
 \frac{\Pr(\aprand[w][i+n]=\apsamp[][i+n])}{\Pr(\aprand[w'][i+n]=\apsamp[][i+n])}\leq \left(\frac{1}{1-p}\right)^{|D_{1,0}^i|-|D_{0,1}^i|}
\end{equation}

On the contrary, if $|D_{0,1}^i|>|D_{1,0}^i|$, using \eqref{eq:GBvsG2} we get
\begin{equation}  \label{eq:dpkw3}
 \frac{\Pr(\aprand[w][i+n]=\apsamp[][i+n])}{\Pr(\aprand[w'][i+n]=\apsamp[][i+n])}\leq \left(\frac{p+q(1-p)}{q}\right)^{|D_{0,1}^i|-|D_{1,0}^i|}
\end{equation}

Now, using the fact that $\sum_{i=1}^{|h|} |D_{0,1}^i|=|D_{0,1}|$ and $\sum_{i=1}^{|h|} |D_{1,0}^i|=|D_{1,0}|$, we write
\begin{equation} \label{eq:dpkw4}
 \prod_{i=1}^{|h|} \frac{\Pr(\aprand[w][i+n]=\apsamp[][i+n])}{\Pr(\aprand[w'][i+n]=\apsamp[][i+n])} \leq \left({\scriptstyle \frac{1}{1-p}}\right)^{|D_{1,0}|} \left({\scriptstyle \frac{p+q(1-p)}{q}}\right)^{|D_{0,1}|}\,.
\end{equation}
Here, we have used the fact that \eqref{eq:dpkw2} and \eqref{eq:dpkw3} are maximized when $D_{0,1}^i$ is empty when $D_{1,0}^i$ is not, and vice-versa. Finally, multiplying \eqref{eq:dpkw1} and \eqref{eq:dpkw4}, we get
\begin{align}
 \frac{\Pr(\mathcal{M}(\mathcal{D},w)=\apsamp)}{\Pr(\mathcal{M}(\mathcal{D},w')=\apsamp)}= &\leq \left({\frac{1}{1-p}\cdot\frac{p+q(1-p)}{q}}\right)^{|D_{0,1}|+|D_{1,0}|}\,,
\end{align}
and therefore according to Definition~\ref{def:dpkeywords},
\begin{equation} \label{eq:epskw}
 \epsilon=\ln\left(1+\frac{p}{q(1-p)}\right)\,,
\end{equation}
which implies Theorem~\ref{theorem:dp-dpsse}.

Note that our $\epsilon$ bounds in \eqref{eq:epsdocs} and \eqref{eq:epskw} can be made smaller by allowing a small probability of failure $\delta$ and using the advanced composition rule \cite{dwork2014algorithmic} of differential privacy.

\subsection{Analysis of $\countermax$ in $\OSSE$.}
\label{app:countermax1}

This analysis is equivalent to the balls-and-bins problem with $\freqmax$ balls and $\freqmax$ bins. We assume $|h|=\freqmax$ and that the hash function outputs values uniformly at random. We sat that the search index construction \emph{succeeds} if we pick a $\countermax$ value that is \emph{strictly larger} (since the counter starts at 0) than the maximum number of documents that share a label and keyword in common. Then, we want to prove that, when $\countermax=c\cdot\ln\freqmax/\ln\ln\freqmax$ (for some constant $c$ that we will determine), the probability of success is overwhelming (larger than $1-1/n$).

Let $S$ denote the success event, and let $n_{i,j}$ be the number of documents that have keyword $w_{(i)}$ and label $j$ ($i\in[|\Delta|], j\in[\freqmax]$). Then,
\begin{align}
 \Pr(S) &=\Pr\left(\cap_{i=1}^{|\Lambda|}\cap_{j=1}^{\freqmax} \{n_{i,j}< \countermax\}\right)\\
        &=1-\Pr\left(\cup_{i=1}^{|\Lambda|} \{\max_j n_{i,j}\geq \countermax\}\right)\\
        &\geq 1-\sum_{i=1}^{|\Lambda|}\cdot\Pr\left(\max_j n_{i,j}\geq \countermax\right)\,.
\end{align}

Now, we bound $\Pr(\max_j n_{i,j}\geq \countermax)$. Note that, for a keyword $w_{(i)}$, the maximum number of documents that have that keyword is $\freqmax$.
\begin{align}
 \Pr(\max_j n_{i,j}\geq \countermax)&\leq{\freqmax \choose \countermax} \left(\frac{1}{\freqmax}\right)^\countermax\\
&\leq\left(\frac{e}{\countermax}\right)^{\countermax}\,.
\end{align}
Now, using that $\countermax=c\cdot\ln\freqmax/\ln\ln\freqmax$,
\begin{align*}
 &\Pr(\max_j n_{i,j}\geq \countermax)\\
\leq& \left(\frac{e}{\countermax}\right)^\countermax\\
  =&\exp\left(\frac{c\cdot\ln\freqmax}{\ln\ln\freqmax}\cdot\ln\frac{e\ln\ln\freqmax}{c\ln\freqmax}\right)\\
  \stackrel{(a)}{\leq}& \exp\left(\frac{c\cdot\ln\freqmax}{\ln\ln\freqmax}\cdot(\ln\ln\ln\freqmax-\ln\ln\freqmax)\right)\\
  =&\exp\left(-c\ln\freqmax+\frac{c\ln\freqmax\cdot\ln\ln\ln\freqmax}{\ln\ln\freqmax}\right)\\
  \leq& \exp(-(c-1)\ln\freqmax) = \frac{1}{\freqmax^{c-1}}\,,
\end{align*}
where $(a)$ holds when $c\geq e$.
Then,
\begin{equation}
 \Pr(S)\geq 1-|\Lambda|\cdot\Pr(\max_j n_{i,j}\geq \countermax)\geq 1-\frac{|\Lambda|}{\freqmax^{c-1}}\,.
\end{equation}
This probability is larger than $1-1/n$ when
\begin{equation}
 c\geq\frac{\ln|\Delta|+\ln n}{\ln\freqmax}+1\,.
\end{equation}
Since $n\gg \freqmax$ and typically $|\Delta|\geq \freqmax$, then the assumption that $c\geq e$ in $(a)$ above holds.

This implies that we need
\begin{align}
 \countermax &= \left(\frac{\ln|\Delta|+\ln n}{\ln\freqmax}+1\right)\cdot \frac{\ln \freqmax}{\ln \ln\freqmax}\\
 &= \frac{\ln |\Delta|+\ln n+\ln \freqmax}{\ln\ln\freqmax}\leq \frac{3\cdot\ln n}{\ln\ln\freqmax}
\end{align}

\subsection{Analysis of Expected Number of Documents with a Keyword $w$ Under Different Distributions}
\label{app:distributions}

We study $E_w\doteq\text{E}\{\mathcal{D}(w)\}$, i.e., the expected number of documents that have a particular keyword $w$, for different keyword and query distributions.

The general expression for $E_w$ is
\begin{equation}
 E_w=\sum_{i=1}^{|\Delta|} \Pr(w_{(i)})\cdot |\mathcal{D}(w)|\,,
\end{equation}
where $\Pr(w_{(i)})$ is the probability that the client queries for keyword $w_{(i)}$.

\subsubsection{Uniform Distribution}
When all keywords have the same (maximum) frequency of appearance $|\mathcal{D}(w)|=\freqmax$ then, regardless off $\Pr(w_{(i)})$, we have
\begin{equation}
 E_w=\sum_{i=1}^{|\Delta|} \Pr(w_{(i)})\cdot \freqmax=\freqmax\,.
\end{equation}

\subsubsection{Zipfian Distribution}

Zipf's law states that the frequency of a individual word in a corpus of natural language utterances is inversely proportional to its rank (the position of it in a sorted list in decreasing order of frequency)~\cite{zipf1932selected}. Assume the keywords in $\Delta$ are sorted in descending frequency order. We know that $|\mathcal{D}(w_{(1)})|=\freqmax$. Then, if keyword frequency follows Zipf's law, we can write
\begin{equation}
 |\mathcal{D}(w_{(i)})|=\frac{\freqmax}{i}\,.
\end{equation}
When query frequencies also follow Zipf's law, we get
\begin{equation}
 \Pr(w_{(i)})=\frac{1}{i\cdot H_{|\Delta|}}\,,\quad\text{ where }H_{|\Delta|}\doteq\sum_{i=1}^{|\Delta|}\frac{1}{i}\,.
\end{equation}

Then,
\begin{equation}
 E_w=\sum_{i=1}^{|\Delta|} \frac{1}{i\cdot H_{|\Delta|}}\cdot \frac{\freqmax}{i}\approx\frac{\freqmax}{\ln{|\Delta|}+\gamma}\cdot \frac{\pi^2}{6}\,,
\end{equation}
where we have used that the harmonic number $H_n\approx \log n + \gamma$ where $\gamma\approx 0.58$ is the Euler-Mascheroni constant, and $\sum_{i=1}^{|\Delta|} i^{-2}\approx \pi^2/6$.

\subsubsection{Worst-Case Distribution}

$\OSSE$'s communication overhead is inversely proportional to $E_w$, so the worst-case distribution for $\OSSE$ is the one that minimizes $E_w$. Since each keyword is in at least one document, and the maximum keyword frequency is $\freqmax$, the worst-case will happen when there is a set of keywords with frequency one and those keywords are the only ones that the client queries, so $E_w\approx 1$.

\subsection{Adapting co-occurrence attacks against $\OSSE$ and $\CLRZ$.}
\label{app:adapt}

In our evaluation, we consider three attacks that use co-occurrence and volume information, namely IKK~\cite{islam2012access}, count attack~\cite{cash2015leakage}, and graph matching~\cite{pouliot2016shadow}.
These attacks compute the co-occurrence matrices $M$ and $M'$ from the observations and the auxiliary (training) information, respectively, and use these matrices to find an assignment of queries to keywords.
We explain how to compute these matrices, and then we delve into specifics of each attack.

First, consider a case where no defense is applied.
Let $m$ be the number of \emph{distinct} obfuscated access patterns observed by the adversary ($m\leq|\Delta|$), and let $\aprand[(i)]$ be the $i$th distinct obfuscated access pattern observed ($i\in[m]$).
Recall that $n$ is the number of documents in the dataset.
Then, $M$ is an $m\times m$ matrix such that its $i,j$th entry contains the \emph{probability} that a document is returned for both the $i$th and $j$th distinct queries.
Mathematically,
\begin{equation} \label{eq:M}
		M[i][j]=\langle \aprand[(i)],\aprand[(j)]\rangle/n\,,
\end{equation}
where $\langle\cdot,\cdot\rangle$ is the dot product.

Matrix $M'$ is a $|\Delta|\times|\Delta|$ matrix whose $i,j$th entry contains the probability that a document contains both $w_{(i)}$ and $w_{(j)}$, estimated using the training set.
Let $\ap[(i)]$ be the true access pattern of keyword $w_{(i)}$ in the training data and let $n'$ be the total number of documents in the training data.
Then,
\begin{equation} \label{eq:Mp_naive}
	M'[i][j]=\langle \ap[(i)],\ap[(j)]\rangle/n'\,.
\end{equation}

In order to adapt the attacks against $\CLRZ$ and $\OSSE$, we modify the computation of $M'$ by taking the true positive and false positive rates ($\TPR$ and $\FPR$) into account.
A document is returned as response for queries $w_{(i)}$ and $w_{(j)}$ with probability $\TPR^2$ when that document contains both keywords, with probability $\TPR\cdot\FPR$ when it only contains one of the keywords, and with probability $\FPR^2$ when it contains neither.
Mathematically, let $n_{i,j}\doteq \langle \ap[(i)],\ap[(j)]\rangle/n'$ be the normalized number of documents that have both $w_{(i)}$ and $w_{(j)}$, and let $\bar{n}_{i,j}\doteq\langle 1-\ap[(i)],1-\ap[(j)]\rangle/n'$ be the normalized number of documents that do not have $w_{(i)}$ nor $w_{(j)}$.
Then, we update the computation of $M'$ as follows:
\begin{equation} \label{eq:Mp}
  M'_{i,j}=\begin{cases} i\neq j:\quad&\TPR^2\cdot n_{i,j} + \FPR^2\cdot \bar{n}_{i,j}\\
	&+ \TPR\cdot\FPR\cdot(1-n_{i,j}-\bar{n}_{i,j})\,,\\
	  i=j:\quad&\TPR\cdot n_{i,i} + \FPR \cdot \bar{n}_{i,i}\,.
		\end{cases}
\end{equation}

All the attacks we evaluate, except for the frequency attack~\cite{liu2014search}, use these matrices.
We explain particularities of each attack below.

\subsubsection{\textbf{Adapting IKK}}
%\label{app:ikk}

The naive version of the attack, that we simply call $\IKK$ in our experiments, uses $M$ and $M'$ as in \eqref{eq:M} and \eqref{eq:Mp_naive}, respectively.
Against $\CLRZ$, we use the original implementation~\cite{islam2012access} where each distinct access pattern is assigned to a unique distinct keyword.
Against $\OSSE$, however, the same keyword can generate different access patterns.
Therefore, in that case we modify the heuristic IKK annealing algorithm so that it allows assigning the same keyword to different observed access patterns.

The improved version of the attack, that we call $\IKKS$ in our experiments, uses $M'$ from \eqref{eq:Mp}.

\subsubsection{\textbf{Adapting the count attack}}
%\label{app:count}

The count attack~\cite{cash2015leakage} first builds a list of candidate keywords for each observed query based on the query response volume and background information about the dataset (in our experiments, we give the adversary the full plaintext database).
Then, the count attack rules out assignments using co-occurrence information until all of the queries are disambiguated and only one possible matching remains.

Cash et al.~propose a generalization of the attack when the background information is imperfect.
In this variation, they assume that the number of documents returned in response to a query follows a Binomial distribution parametrized by the background (training) information, and use confidence intervals derived from Hoeffding bounds to build keyword candidates and disambiguate queries.
This version of the count attack applies naturally against $\CLRZ$ and $\OSSE$, since adding false positives and false negatives to the access patterns actually causes the query volume and co-occurrences to follow a Binomial distribution.

The attack against $\CLRZ$ proceeds as follows:
\begin{enumerate}
\item Compute the co-occurrence matrices $M$ and $M'$ as in \eqref{eq:M} and \eqref{eq:Mp}, respectively.
\item Compute the confidence intervals such that the observed volumes and co-occurences (in $M$) are within that interval (centered around the values in $M'$) with probability $p=0.95$ (c.f.~\cite{cash2015leakage}, Theorem 4.1).
\item Build the set of candidate keywords for each observed query based on the volume information (diagonal elements of $M$ and $M'$).
If any candidate set is empty, the attack fails.
\item Disambiguate queries by following the brute-force approach by Cash et al.~(c.f.~\cite{cash2015leakage}, Section 4.4). 
We first select the set of the 10 most-frequently observed queries, build all the possible assignments of those queries to their candidate keywords, and try to disambiguate queries for each of those assignments using co-occurrence information.
The attack returns the assignment that disambiguated the highest number of queries.
If \emph{all} of the assignments resulted in an inconsistency (the candidate set of a keyword became empty), the attack fails.
\end{enumerate}

Running the count attack against $\OSSE$ is not straightforward, since the attack heuristics assume that each distinct query belongs to a different keyword.
In $\OSSE$, however, the same keyword will likely generate different access patterns every time it is queried.
Therefore, we apply the clustering technique that we used for the frequency attack in Section~\ref{sec:freq}: the adversary takes all the observed access patterns and clusters them into $N_c$ groups, where $N_c$ is the true number of distinct queries issued by the client (we give the adversary this ground-truth information so that its performance is a worst-case against $\OSSE$).
Then, the adversary takes the center of each cluster as the ``representative'' access pattern of that group (note that this is still consistent with \eqref{eq:Mp}) and runs the attack as explained above.
All the elements within a cluster are assigned the keyword of their representative access pattern.

\subsubsection{\textbf{Adapting the graph matching attack}}

The graph matching attack~\cite{pouliot2016shadow} uses keyword co-occurrence information, similar to IKK.
We evaluate the adapted attack only, i.e., using $M$ and $M'$ from \eqref{eq:M} and \eqref{eq:Mp}, respectively.
Since this attack also assumes that each distinct access pattern is assigned to a different keyword, we use the clustering technique to group access patterns in $\OSSE$ and then run the graph matching algorithm using the cluster centers.

\end{document}